\input ./style/arxiv-general.cfg
\documentclass[MSNbibl,nameyear,seceqn,dvips]{arxstspdf}
\makeatletter
   \@ifpackageloaded{graphicx}{}{\usepackage{graphicx}}
\makeatother
\usepackage{mathbh}
\usepackage{stfloats}

\setattribute{abstract}   {width}  {35pc}
\setattribute{keyword}    {width}  {35pc}

\volume{30}
\issue{4}
\pubyear{2015}
\firstpage{443}
\lastpage{467}
\doi{10.1214/15-STS523}
\docsubty{FLA}

\makeatletter
  \fnbelowfloat
\def\sfrac#1#2{#1/#2}

\newcommand{\rrvert}{\vert}
\newcommand{\llvert}{\vert}
\renewcommand{\mid}{|}
\newcommand{\eqref}[1]{(\ref{#1})}
\renewcommand{\citep}[1]{(\citeauthor{#1}, \citeyear{#1})}
\newcommand{\tr}{\operatorname{tr}}
\def\rset{\mathbb R}
\def\Xset{\mathsf{X}} 
\def\PP{\mathbb{P}} 
\def\PE{\mathbb{E}} 
\newcommand{\B}{\mathcal{B}}
\newcommand{\eqdef}{\stackrel{\mathrm{def}}{=}}
\newtheorem{proposition}{Proposition}[section]
\newproclaim{remark}{Remark}[section]
\makeatother

\begin{document}
\begin{frontmatter}

\title{On Russian Roulette Estimates for Bayesian Inference with Doubly-Intractable~Likelihoods\thanksref{T11}}
\runtitle{Playing Russian Roulette with Doubly-Intractable Likelihoods}
\thankstext{T11}{Code to replicate all results reported can be
downloaded from \surl{http://www.ucl.ac.uk/roulette}.}

\begin{aug}
\author[A]{\fnms{Anne-Marie}~\snm{Lyne}\corref{}\ead[label=e1]{lyne.annemarie3@gmail.com}},
\author[B]{\fnms{Mark}~\snm{Girolami}\ead[label=e2]{m.girolami@warwick.ac.uk}},
\author[C]{\fnms{Yves}~\snm{Atchad\'e}\ead[label=e3]{yvesa@umich.edu}},
\author[D]{\fnms{Heiko}~\snm{Strathmann}\ead[label=e4]{heiko.strathmann@gmail.com}}
\and
\author[E]{\fnms{Daniel}~\snm{Simpson}\ead[label=e5]{dp.simpson@gmail.com}}
\runauthor{A.-M. Lyne et al.}

\affiliation{University College London,
University of Warwick,
University of Michigan,
University College London,
University of Warwick}

\address[A]{Anne-Marie Lyne is a Ph.D. student, Department of Statistical Science,
University College London, London WC1E 6BT, United Kingdom  \printead{e1}.}
\address[B]{Mark Girolami is Professor, Department of Statistics, University of Warwick, Coventry CV4 7AL, United Kingdom \printead{e2}.}
\address[C]{Yves Atchad\'e is Associate Professor, Department of Statistics, University of Michigan, Ann Arbor, Michigan 48109, USA \printead{e3}.}
\address[D]{Heiko Strathmann is a Ph.D. student, Gatsby Computational Neuroscience Unit, University College London, London W1T 4JG, United Kingdom \printead{e4}.}
\address[E]{Daniel Simpson is a CRiSM Fellow, Department of Statistics, University of Warwick, Coventry CV4 7AL, United Kingdom \printead{e5}.}
\end{aug}

%
\begin{abstract}
A large number of statistical models are ``doubly-intractable'': the
likelihood normalising
term, which is a function of the model parameters, is intractable, as
well as the marginal likelihood (model evidence). This means that
standard inference techniques to sample from the posterior, such as
Markov chain Monte Carlo (MCMC), cannot be used. Examples include, but
are not confined to, massive Gaussian Markov random fields,
autologistic models and Exponential random graph models. A number of
approximate schemes based on MCMC techniques, Approximate Bayesian
computation (ABC) or analytic approximations to the posterior have been
suggested, and these are reviewed here. Exact MCMC schemes, which can
be applied to a subset of doubly-intractable distributions, have also
been developed and are described in this paper. As yet, no general
method exists which can be applied to all classes of models with
doubly-intractable posteriors.

In addition, taking inspiration from the Physics literature, we study
an alternative
method based on representing the intractable likelihood as an infinite
series. Unbiased estimates of the likelihood can then be obtained by
finite time stochastic truncation of the series via Russian Roulette
sampling, although the estimates are not necessarily positive. Results
from the Quantum Chromodynamics literature are exploited to allow the
use of possibly negative estimates in a pseudo-marginal MCMC scheme
such that expectations with respect to the posterior distribution are
preserved. The methodology is reviewed on well-known examples such as
the parameters in Ising models, the posterior for Fisher--Bingham
distributions on the $d$-Sphere and a large-scale Gaussian Markov
Random Field model describing the Ozone Column data. This leads to a
critical assessment of the strengths and weaknesses of the methodology
with pointers to ongoing research.

\end{abstract}

%
\end{frontmatter}

\setcounter{footnote}{1}

\section{Introduction}\label{sec1}

An open problem of growing importance in the application of Markov
chain Monte Carlo (MCMC) methods for Bayesian computation is the
definition of transition kernels for posterior distributions with
intractable data densities. In this paper, we focus on methods for a
subset of these distributions known as \emph{doubly-intractable}
distributions, a term first coined by \citet{murray2006mcmc}. To
illustrate what constitutes a doubly-intractable posterior, take some
data ${\mathbf y} \in{\mathcal Y}$ used to make posterior inferences about
the variables ${\bolds\theta} \in{\bolds\Theta}$ that
define a statistical model. A prior distribution defined by a density
$\pi({\bolds\theta})$ with respect to Lebesgue measure $d{\bolds\theta}$ is adopted and the data density is given by
$p({\mathbf y}\mid{\bolds\theta}) = f({\mathbf y};{\bolds\theta
})/{\mathcal Z}({\bolds\theta})$, where $f({\mathbf y};{\bolds
\theta})$ is an unnormalised function of the data and parameters, and
${\mathcal Z}({\bolds\theta}) = \int f({\mathbf x};{\bolds
\theta})\,d{\mathbf x}$ is the likelihood normalising term which
\emph{cannot be computed}. The posterior density follows in the usual
form as
%
%
\begin{eqnarray}\label{eqnposterior}
\pi({\bolds\theta}\mid{\mathbf y}) &=& \frac
{p({\mathbf y}\mid{\bolds
\theta}) \times\pi({\bolds\theta})}{
p({\mathbf y})}
\nonumber\\[-8pt]\\[-8pt]\nonumber
&=&
\frac{f({\mathbf y};{\bolds\theta})}{{\mathcal Z}({\bolds\theta})} \times\pi({\bolds\theta}) \times\frac{1}{p({\mathbf y})},
\end{eqnarray}
where $p({\mathbf y}) = \int p({\mathbf y}\mid{\bolds\theta}) \pi
({\bolds\theta}) \,d{\bolds\theta}$. ``Doubly-intractab\-le''
refers to the fact that not only is $p({\mathbf y})$ intractable (this is
common in Bayesian inference and does not generally present a problem
for inference), but ${\mathcal Z}({\bolds\theta})$ is also intractable.

Bayesian inference proceeds by taking posterior expectations of
functions of interest, that is,
%
%
\begin{equation}
E_{\pi({\bolds\theta}\mid{\mathbf y})} \bigl\{h({\bolds\theta})
\bigr\} = \int h({\bolds\theta})\pi({\bolds\theta}|{\mathbf y}) \,d{\bolds\theta}
\end{equation}
and Monte Carlo estimates of the above expectations can be obtained by
employing MCMC methods if other exact sampling methods are not
available (\cite{Gilks99}, \cite{RRobert+Casella2010}, \cite{Liu01}, \cite{Gelman03}). To
construct a Markov chain with invariant distribution $\pi({\bolds\theta}|{\mathbf y})$, the Metropolis--Hastings algorithm can be used; a
transition kernel is constructed by designing a proposal distribution
$q({\bolds\theta}'|{\bolds\theta})$ and accepting the
proposed parameter value with probability
%
%
\begin{eqnarray}
\alpha\bigl({\bolds\theta}', {\bolds\theta}\bigr) &=& \min
\biggl\{ 1, \frac{\pi({\bolds\theta'}|{\mathbf y}) q({\bolds\theta
}|{\bolds\theta}')}{\pi({\bolds\theta}|{\mathbf y})
q({\bolds\theta}'|{\bolds\theta})} \biggr\}
\nonumber\\[-8pt]\\[-8pt]\nonumber
&=& \min\biggl\{1,
\frac{ f({\mathbf y};{\bolds\theta'})\pi({\bolds\theta
}')q({\bolds\theta}'|{\bolds\theta}) }{ f({\mathbf
y};{\bolds\theta}) \pi({\bolds\theta}) q({\bolds\theta}'|{\bolds\theta})} \times
\frac{ {\mathcal
Z}({\bolds\theta}) }{ {\mathcal Z}({\bolds\theta}') } \biggr\}.\hspace*{-15pt}
\end{eqnarray}

Clearly, a problem arises when the value of the normalising term for
the data density, ${\mathcal Z}({\bolds\theta})$, cannot be
obtained either due to it being nonanalytic or uncomputable with a
finite computational resource. This situation is far more widespread in
modern-day statistical applications than a cursory review of the
literature would suggest and forms a major challenge to methodology for
computational statistics currently (e.g., \cite{moller2006efficient},
\cite{besag1975estimation},
\cite{besag1974spatial},
\cite{green2002hidden},
\cite{moller2003statistical}).
We review and study methods which have been published in the Statistics
and Physics literature for dealing with such distributions. We then
study in detail how to implement a pseudo-marginal MCMC scheme
(\cite{beaumont2003estimation}, \cite{andrieu2009pseudo}) in which an unbiased
estimate of the target density is required at each iteration, and
suggest how these might be realised.

This paper is organised as follows. In Section~\ref{secreview} we
describe examples of doubly-intractable distributions along with
current inference approaches. These encompass both approximate and
exact methods which have been developed in the Statistics, Epidemiology
and Image analysis literature. In Section~\ref{secPM} we suggest an
alternative approach based on pseudo-marginal MCMC
(\cite{beaumont2003estimation}, \cite{andrieu2009pseudo}) in which an unbiased
estimate of the intractable target distribution is used in an MCMC
scheme to sample from the exact posterior distribution. In
Sections~\ref{secPMdoub} and \ref{secunbiased} we describe how to
realise such unbiased estimates of a likelihood with an intractable
normalising term. This is achieved by writing the likelihood as an
infinite series in which each term can be estimated unbiasedly. Then
Russian Roulette techniques are used to truncate the series such that
only a finite number of terms need be estimated whilst maintaining the
unbiasedness of the overall estimate. Sections~\ref{secexp} and \ref
{secOzone} contain experimental results for posterior inference over
doubly-intractable distributions: Ising models, the Fisher--Bingham
distribution and a large-scale Gaussian Markov random field.
Section~\ref{secdisc} contains a discussion of the method and suggests
areas for further work.

\section{Inference Methods for Doubly-Intractable Distributions}\label{secreview}

\subsection{Approximate Bayesian Inference}

Many models describing data with complex dependency structures are
doubly-intractable. Examples which have received attention in the
Statistics literature include:

%
\begin{longlist}[2.]
\item[1.] The Ising model \citep{ising1925beitrag}. Originally formulated
in the Physics literature as a simple model for interacting magnetic
spins on a lattice. Spins are binary random variables which interact
with neighbouring spins.

\item[2.] The Potts model and autologistic models. Generalisations to the
Ising model in which spins can take more than two values and more
complex dependencies are introduced. These models are used in image
analysis (\cite{besag1986statistical}, \cite{hughes2011autologistic}), as well
as in other fields such as disease mapping
(e.g., \cite{green2002hidden}).
\item[3.] Spatial point processes. Used to model point pattern data, for
example, ecological data (e.g., \cite{silvertown2001integrating}, 
M{\o}ller and Waage\-petersen (\citeyear{moller2003statistical})) and
epidemiological data (e.g., \cite{diggle1990point}).
\item[4.] Exponential Random Graph (ERG) models. Used in the field of
social networks to analyse global network structures in terms of local
graph statistics such as the number of triangles (e.g., \cite{goodreau2009birds}).
\item[5.] Massive Gaussian Markov random field\break (GMRF) models. Used in image
analysis and spatial statistics, amongst others (e.g., \cite{RueGMRFBook}).
\end{longlist}
%

Standard Bayesian inference techniques such as drawing samples from the
posterior using MCMC cannot be used due to the intractability of the
likelihood normalising term, and hence a number of approximate
inference methods have been developed. A~common approach when the full
likelihood cannot be computed is to use a pseudo-likelihood (\cite{besag1974spatial}, \cite{besag1975estimation}), in which an approximation to
the true likelihood is formed using the product of the conditional
probabilities for each variable. This can normally be computed
efficiently and can therefore \mbox{replace} the full likelihood in an
otherwise standard inference strategy to sample from the posterior
(e.g., \cite{heikkinen1994fully}, \cite{zhou2009bayesian}). This approach
scales well with the size of the data and can give a reasonable
approximation to the true posterior, but inferences may be
significantly biased as long range interactions are not taken into
account [this has been shown to be the case for ERG models \citep
{van2009framework}, hidden Markov random fields \citep
{friel2009bayesian} and autologistic models \citep
{friel2004likelihood}]. Methods based on composite likelihoods have
also been used for inference in massive scale GMRF models, in which an
approximation to the likelihood is based on the joint density of
spatially adjacent blocks \citep{eidsvik2013estimation}. This has the
advantage that the separate parts of the likelihood cannot only be
computed more efficiently, but also computed in parallel.

Another pragmatic approach is that of \citet{green2002hidden}, in
which they discretise the interaction parameter in the Potts model to a
grid of closely spaced points and then set a prior over these values.
Estimates of the normalising term are then precomputed using
thermodynamic integration (as described by \cite{gelman1998simulating}) so that no expensive computation is
required during the MCMC run. This allowed inference to be carried out
over a model for which it would not otherwise have been possible.
However, it is not clear what impact this discretisation and use of
approximate normalising terms has on parameter inference and it seems
preferable, if possible, to retain the continuous nature of the
variable and to not use approximations unless justified.

Approximate Bayesian Computation (ABC)\break \citep{marinabc2012}, a
technique developed for likelihood free inference (\cite{tavare1997inferring}, \cite{beaumont2002approximate}), can also be used. The
types of models for which ABC was originally developed are implicit,
meaning data can be simulated from the likelihood but the likelihood
cannot be written down, and hence neither standard maximum likelihood
nor Bayesian methods can be used. For doubly-intractable distributions,
it is only the normalising term which cannot be computed, but we can
still use the techniques developed in the ABC community. ABC in its
simplest form proceeds by proposing an approximate sample from the
joint distribution, $p(\mathbf y, {\bolds\theta})$, by first
proposing ${\bolds\theta'}$ from the prior and then generating a
data set from the model likelihood conditional on $\bolds{\theta}'$. This data
set is then compared to the observed data and the proposed parameter
value accepted if the generated data is ``similar'' enough to the
observed data. An obvious drawback to the method is that it does not
sample from the exact posterior, although it has been shown to produce
comparable results to other approximate methods and recent advances
mean that it can be scaled up to very large data sets
(\cite{grelaud2009abc}, \cite{everitt2012bayesian}, \cite{moores2014pre}).

The ``shadow prior'' method of \citet{liechty2009shadow} is an
interesting attempt to reduce the computational burden of intractable
normalising constants in the case where constraints on the data or
parameters cause the intractability. As an example, take data ${\mathbf
y}\sim p({\mathbf y}|{\bolds\theta})$ which is constrained to lie in
some set $A$. Depending on the form of $A$, sampling from the posterior
$\pi({\bolds\theta}|{\mathbf y})$ can be hindered by an intractable
likelihood normalising term. The model is therefore replaced by $p({\mathbf
y}|{\bolds\delta})I({\mathbf y}\in A)$, ``shadow prior''
$p({\bolds\delta}|{\bolds\theta}) = \prod_{i =
1}^{d}\mathcal{N}(\delta_i;\theta_i,\nu)$ and prior $\pi
({\bolds\theta})$, for some $\nu$ and where $d$ is the
dimension of $\bolds\theta$. The conditional posterior
$p({\bolds\theta}|{\bolds\delta},{\mathbf y}) = p(\bolds\theta|\bolds\delta)$ no longer requires the computation of an
intractable normalising term (as dependence on the constrained data has
been removed), although updating $\bolds\delta$ does. However,
this has been reduced to $d$ one-dimensional problems which may be
simpler to deal with. The method, of course, only works if the
computational burden of the intractable normalising constant is
significantly less in the shadow prior format than in the original
model, and several examples of when this might be the case are
suggested, such as when the parameter in the normalising constant has a
complicated hyperprior structure. An approximate version can be
implemented in which the normalising constant is ignored in the shadow
prior, which can sometimes have very little impact on the final
inference. In these cases the computational burden has been eliminated.

Several approximate but consistent algorithms have been developed based
on Monte Carlo approximations within MCMC methods. For example, an
approach was developed by \citet{atchade2008bayesian} in which a
sequence of transition kernels are constructed using a consistent
estimate of ${\mathcal Z}(\bolds\theta)$ from the Wang--Landau
algorithm \citep{wang2001efficient}. The estimates of the normalising
term converge to the true value as the number of iterations increases
and the overall algorithm gives a consistent approximation to the
posterior. Bayesian Stochastic Approximation Monte Carlo \citep
{jin2014use} works in a similar fashion, sampling from a series of
approximations to the posterior using the stochastic approximation
Monte Carlo algorithm \citep{liang2007stochastic}, which is based on
the Wang--Landau algorithm. These algorithms avoid the need to sample
from the model likelihood, but in practice suffer from the curse of
dimensionality as the quality of the importance sampling estimate
depends on the number and location of the grid points. These points
need to grow exponentially with the dimension of the space limiting the
applicability of this methodology. They also require a significant
amount of tuning to attain good approximations to the normalising term,
and hence ensure convergence is achieved.

Alternative methodologies have avoided sampling altogether and instead
used deterministic approximations to the posterior distribution. This
is particularly the case for GMRF models which often have complex
parameter dependencies and are very large in scale, rendering MCMC
difficult to apply. INLA (\mbox{integrated} nested Laplace approximations)
\citep{rue2009approximate} was designed to analyse latent Gaussian
models and has been applied to massive GMRFs in diverse areas such as
spatio-temporal disease mapping \citep{schrodle2011spatio} and point
processes describing the locations of muskoxen \citep
{illian2010fitting}. By using Laplace approximations to the posterior
and an efficient programming implementation, fast Bayesian inference
can be carried out for large models. However, this benefit also
constitutes a drawback in that users must rely on standard software,
and therefore model extensions which could be tested simply when using
an MCMC approach are not easy to handle. Further, it is of course
necessary to ensure that the assumptions inherent in the method apply
so that the approximations used are accurate. It should also be noted
that the work of \citet{taylor2013inla} found that in the case of
spatial prediction for log-Gaussian Cox processes, an MCMC method using
the Metropolis-adjusted Langevin Algorithm (MALA) algorithm gave
comparable results in terms of predictive accuracy and was actually
slightly more efficient than the INLA method. Other approximations have
also been developed as part of a large body of work in the area, such
as iterative methods for approximating the log determinant of large
sparse matrices, required to compute the likelihood \citep{aune2014parameter}.

\subsection{Exact MCMC Methods}

As well as approximate inference methods, a small number of exact
algorithms have been developed to sample from doubly-intractable
posteriors. These are described below as well as advice as to when
these algorithms can be used.

\subsubsection{Introducing auxiliary variables}
An exact  sampling methodology for doubly-intractable distributions is
proposed in \citet{walker2011posterior}, which uses a similar approach
to those described in \citet{adams2009nonparametric} and Section~9 of
\citet{Beskos2006id}. A Reversible-Jump MCMC (RJMCMC) sampling scheme
is developed that cleverly gets around the intractable nature of the
normalising term. Consider the univariate distribution
$p(y|{\bolds\theta}) = f(y;{\bolds\theta})/{\mathcal
Z}({\bolds\theta}) $ where $N$ i.i.d. observations, $y_i$, are
available. In its most general form, it is required that $y$ belongs to
some bounded interval $[a, b]$, and that there exists a constant
$M<+\infty$ such that $f(y;{\bolds\theta})<M$ for all
${\bolds\theta}$ and $y$ (it is assumed that $[a,b]=[0,1]$, and
$M = 1$ in the following exposition). The method introduces auxiliary
variables $\nu\in(0,\infty)$, $k\in\{0,1,\ldots\}$, $\{s\}
^{(k)}=(s_1,\ldots,s_k)$, to form the joint density
\begin{eqnarray*}
&& f\bigl(\nu,k,\{s\}^{(k)},{\mathbf y}|{\bolds\theta}\bigr)
\\
&&\quad \propto
\frac{\exp
(-\nu)\nu^{k+N-1}}{k!}
\\
&&\qquad{}\cdot \prod_{j = 1}^{k}
\bigl(1-f(s_{j};{\bolds\theta}) \bigr){\mathbh{1}}(0<s_{j}<1)
\\
&&\qquad{}\cdot \prod_{i = 1}^{N} f(y_{i};\theta).
\end{eqnarray*}
Integrating out $\nu$ and $ s^{(k)}$ and summing over all $k$ returns
the data distribution $\prod_{i=1}^N p(y_i|{\bolds\theta})$. An
RJMCMC scheme is proposed\vspace*{1pt} to sample from the joint density $f(\nu,k,\{
s\}^{(k)},{\mathbf y}|{\bolds\theta})$ and this successfully gets
around the intractable nature of the normalising term. The scheme has
been used to sample from the posterior of a Bingham distribution \citep
{walker2014bayesian}.

However, the methodology has some limitations to its generality.
Firstly, the unnormalised density function must be strictly bounded
from above to ensure the positivity of the terms in the first product.
This obviously limits the generality of the methodology to the class of
strictly bounded functions; however, this is not overly restrictive, as
many functional forms for $f(y_{i};{\bolds\theta})$ are bounded,
for example, when there is finite support, or when
$f(y_{i};{\bolds\theta})$ takes an exponential form with
strictly negative argument. Even if the function to be sampled is
bounded, finding bounds that are tight is extremely difficult and the
choice of the bound directly impacts the efficiency of the sampling
scheme constructed; see, for example, \citet{el2008convex} for bounds
on binary lattice models. Ideally we would wish to relax the
requirement for the data, ${\mathbf y}$, to belong to a bounded interval,
but if we integrate with respect to each $s_j$ over an unbounded
interval, then we can no longer return $1-\mathcal{Z}({\bolds\theta})$ and the sum over $k$ will therefore no longer define a
convergent geometric series equaling ${\mathcal{Z}({\bolds\theta
})}$. This last requirement particularly restricts the generality and
further use of this specific sampling method for intractable distributions.

\subsection{Valid Metropolis--Hastings-Type Transition Kernels}

An ingenious MCMC solution to the doubly-in\-tractable problem was
proposed by \citet{moller2006efficient} in which the posterior state
space is extended as follows:
\begin{eqnarray*}
\pi({\bolds\theta},{\mathbf x}|{\mathbf y}) \propto p({\mathbf x}|{\bolds\theta},{
\mathbf y}) \pi({\bolds\theta}) \frac
{f({\mathbf y};{\bolds\theta})}{{\mathcal Z}(\theta)}.
\end{eqnarray*}

This extended distribution retains the posterior as a marginal. The
method proceeds by taking the proposal for ${\mathbf x},\bolds\theta
$ to be $q({\mathbf x'},{\bolds\theta'}|{\mathbf x},\bolds\theta)
= \frac{f({\mathbf x};{\bolds\theta}')}{{\mathcal Z}({\bolds\theta}')}q({\bolds\theta}'|{\bolds\theta})$, so that at
each iteration the intractable normalising terms cancel in the
Metropolis--Hastings acceptance ratio. A drawback of the algorithm is
the need to choose the marginal for ${\mathbf x}$, $p({\mathbf x}|{\bolds\theta},{\mathbf y})$, particularly as the authors suggest that ideally
this distribution would approximate the likelihood, thereby
reintroducing the intractable normalising term.

\citet{murray2006mcmc} simplified and extended the algorithm to the
Exchange algorithm, and in the process removed this difficulty by
defining a joint distribution as follows:
\begin{eqnarray*}
p\bigl({\mathbf x},{\mathbf y},{\bolds\theta},{\bolds\theta}'
\bigr) \propto\frac{f({\mathbf y};{\bolds\theta})}{{\mathcal
Z}({\bolds\theta})} \pi({\bolds\theta}) q\bigl({\bolds\theta}'|{\bolds\theta}\bigr)\frac{f({\mathbf x};{\bolds\theta
}')}{{\mathcal Z}({\bolds\theta}')}.
\end{eqnarray*}

At each iteration, MCMC proceeds by first Gibbs sampling $\bolds\theta'$ and ${\mathbf x}$, and then proposing to swap the values of
${\bolds\theta}$ and ${\bolds\theta}'$ using
Metropolis--Hastings. Again, the intractable normalising terms cancel in
the acceptance ratio. Both of these algorithms use only valid MCMC
moves and therefore target the \emph{exact} posterior, rendering them
a major methodological step forward. However, they both require the
capability to sample from the likelihood using a method such as perfect
sampling (\cite{propp1996exact}, \cite{kendall2005notes}). This can be
considered a restriction to the widespread applicability of this class
of methods, as for many models it is not possible, for example, the ERG
model in social networks. Even when perfect sampling is possible, for
example, for the Ising and Potts models, it becomes prohibitively slow
as the size of the model increases. Attempts have been made to relax
the requirement to perfectly sample by instead using an auxiliary
Markov chain to sample approximately from the model at each iteration
(\cite{caimo2011bayesian},
\cite{liang2010double},
\cite{everitt2012bayesian},
\cite{alquier2014noisy}).
In particular, the paper by \citet{alquier2014noisy} suggests multiple
approximate MCMC algorithms for doubly-intractable distributions and
then applies results from Markov chain theory to bound the total
variation distance between the approximate chains and a hypothetical
exact chain. These types of approximate algorithms were in use due to
their computational feasibility, and so it is pleasing to see some
theoretical justification for their use emerging in the Statistics literature.

\section{An Alternative Approach Using Pseudo-Marginal MCMC}\label{secPM}

As has been seen, there are many approximate methods for sampling from
doubly-intractable posteriors. There are also exact methods available,
but these can only be applied when it is possible to perfectly sample
from the data model. Now we would like to approach the question of
whether it is possible to relax this requirement and develop
methodology for exact sampling of the posterior when perfect sampling
is not possible. To do this, we develop an approach based on the
pseudo-marginal methodology
(\cite{beaumont2003estimation},
\cite{andrieu2009pseudo},
\cite{doucet2012efficient}), and
hence we now briefly review the algorithm. The pseudo-marginal class of
methods is particularly appealing in that they have the least number of
restrictions placed upon them and provide the most general and
extensible MCMC methods for intractable \mbox{distributions}. They are
sometimes referred to as Exact-approximate methods, based on the
property that the invariant distribution of the Markov chain produced
is the exact target distribution despite the use of an approximation in
the Metropolis--Hastings acceptance probability. To use the scheme, an
unbiased and positive estimate of the target density is substituted for
the true density, giving an acceptance probability of the form
%
%
\begin{eqnarray}\label{acceptprob}
\alpha\bigl({\bolds\theta}', {\bolds\theta}
\bigr) &=& \min\biggl\{ 1, \frac{\hat\pi({\bolds\theta}'|{\mathbf
y})}{\hat\pi
({\bolds\theta}|{\mathbf y})} \times\frac{q({\bolds\theta} |
{\bolds\theta
}')}{q({\bolds\theta}' | {\bolds\theta})} \biggr\}
\nonumber\\[-8pt]\\[-8pt]\nonumber
&=&
\min\biggl\{1, \frac{\hat{p}({\mathbf y}|{\bolds\theta}')\pi
({\bolds\theta
}')}{\hat{p}({\mathbf y}|{\bolds\theta})\pi({\bolds\theta})}
\times\frac{q({\bolds\theta} | {\bolds\theta
}')}{q({\bolds\theta}' | {\bolds\theta})} \biggr\},
\end{eqnarray}
where the estimate at each proposal is propagated forward as described
in \citet{beaumont2003estimation}, \citet{andrieu2009pseudo}. For the case of
doubly-intractable distributions, assuming the prior is tractable, this
equates to a requirement for an unbiased estimate of the likelihood as
seen on the right in \eqref{acceptprob} above. The remarkable feature
of this scheme is that the corresponding transition kernel has an
invariant distribution with ${\bolds\theta}$-marginal given
precisely by the desired posterior distribution, $\pi({\bolds\theta}|{\mathbf y})$. To see this, denote all the random variables
generated in the construction of the likelihood estimator by the vector
$\mathbf u$ and its density $p(\mathbf u)$. These random variables are, for
example, those used when generating and accepting a proposal value in a
Markov chain as part of a Sequential Monte Carlo estimate. The
estimator of the likelihood is denoted $\hat{p}_N(\mathbf y|{\bolds\theta},\mathbf{u})$, with $N$ symbolising, for example, the number of Monte
Carlo samples used in the estimate. The estimator of the likelihood
must be unbiased, that is,
%
%
\begin{eqnarray}
\label{unbiased} \int\hat{p}_N({\mathbf y}|{\bolds\theta},{\mathbf u}) p({
\mathbf u}) \,d {\mathbf u} = p({\mathbf y}|{\bolds\theta}).
\end{eqnarray}

A joint density for ${\bolds\theta}$ and $\mathbf u$ is now defined
which returns the posterior distribution after integrating over~$\mathbf u$:
\begin{eqnarray*}
\pi_N({\bolds\theta},{\mathbf u}|{\mathbf y}) & \propto&
\hat{p}_N({\mathbf y}|{\bolds\theta},{\mathbf u}) \pi({\bolds\theta})p({\mathbf u})
\\
& =& \frac{\hat{p}_N({\mathbf y}|{\bolds\theta},{\mathbf u}) \pi
({\bolds\theta})p({\mathbf u})}{p({\mathbf y})}.
\end{eqnarray*}

It is simple to show using equation~\eqref{unbiased} that $\pi
_N({\bolds\theta},\mathbf u|\mathbf y)$ integrates to $1$ and has the
desired marginal distribution for ${\bolds\theta}|\mathbf y$. Now
consider sampling from $\pi_N({\bolds\theta},\mathbf u|\mathbf y)$
using the Metropolis--Hastings algorithm, with the proposal distribution
for $\mathbf u'$ being $p(\mathbf u')$. In this case the densities for $\mathbf u$
and $\mathbf u'$ cancel and we are using the acceptance probability in
\eqref{acceptprob}. Hence, this algorithm samples from $\pi
_N({\bolds\theta},\mathbf u|\mathbf y)$ and the samples of ${\bolds\theta}$ obtained are distributed according to the posterior.

This is a result that was highlighted in the statistical genetics
literature \citep{beaumont2003estimation}, then popularised and
formally analysed in \citet{andrieu2009pseudo} with important
developments such as Particle MCMC \citep{doucet2012efficient} proving
to be extremely powerful and useful in a large class of statistical
models. Due to its wide applicability, the pseudo-marginal algorithm
has been the subject of several \mbox{recent} papers in the statistical
literature, increasing understanding of the methodology. These have
covered how to select the number of samples in the unbiased estimate to
minimise the computational time \citep{doucet2012efficient}, optimal
variance and acceptance rates to maximise efficiency of the chain
\citep{sherlock2013efficiency} and results to order two different
pseudo-marginal implementations in terms of the acceptance probability
and asymptotic variance \citep{andrieu2014establishing}. It is
interesting to note that the problem of Exact-Approximate inference was
first considered in the Quantum Chromodynamics literature almost thirty
years ago. This was motivated by the need to reduce the computational
effort of obtaining values for the strength of bosonic fields in
defining a Markov process to simulate configurations following a
specific law; see, for example
\citet{kennedy1985noise},
\citet{bhanot1985bosonic},
\citet{bakeyev2001noisy},
\citet{lin2000noisy},
\citet{joo2003kentucky}.


\subsection{Proposed Methodology}

One can exploit the pseudo-marginal algorithm to sample from the
posterior, and hence we require unbiased estimates of the likelihood.
For each $\bolds\theta$~and~$\mathbf{y}$, we\vspace*{1pt} show that one
can construct random variables $\{V^{(j)}_{\bolds\theta}, j\geq
0\}$ (where dependence on ${\mathbf y}$ is omitted) such that the series
defined as
\[
\pi\bigl(\bolds\theta,\bigl\{V^{(j)}_{{\bolds\theta}}\bigr\}
|{\mathbf y}
\bigr):= \sum_{j= 0}^\infty
V^{(j)}_{{\bolds\theta}}
\]
is finite almost surely, has finite expectation, and $\mathbb{E}
(\pi({\bolds\theta},\{V^{(j)}_{\bolds\theta}\} |
{\mathbf y}) )=\pi({\bolds\theta}|{\mathbf y})$. We propose a
number of ways to construct such series. Although unbiased, these
estimators are not practical, as they involve infinite series. We
therefore employ a computationally feasible truncation of the infinite
sum which, crucially, remains unbiased. This is achieved using Russian Roulette procedures well known in the Physics \mbox{literature}
(\cite{hendricks1985mcnp}, \cite{carter1975particle}). More precisely, we introduce a
random time $\tau_{{\bolds\theta}}$, such that with ${\mathbf
u}:=(\tau_{{\bolds\theta}},\{V^{(j)}_{\bolds\theta},
0\leq j\leq\tau_{{\bolds\theta}}\})$ the estimate
\begin{eqnarray*}
\pi({\bolds\theta},{\mathbf u}|{\mathbf y}) := \sum_{j= 0}^{\tau
_{\bolds\theta}}
V^{(j)}_{\bolds\theta}
\end{eqnarray*}
satisfies
\begin{eqnarray*}
 \mathbb{E} \bigl(\pi({
\bolds\theta},{\mathbf u}|{\mathbf y}) |\bigl\{V^{(j)}_{\bolds\theta}, j
\geq0\bigr\} \bigr)
=\sum_{j= 0}^\infty
V^{(j)}_{{\bolds\theta}}.
\end{eqnarray*}

As in the notation used above, $\mathbf u$ is a vector of all the random
variables used in the unbiased estimate, that is, those used to
estimate terms in the series, as well as those used in the roulette
methods to truncate the series. As the posterior is only required up to
a normalising constant in $\mathbf y$ and the prior is assumed tractable,
in reality we require an unbiased estimate of the likelihood.

\subsection{The Sign Problem}

If the known function $f({\mathbf y};{\bolds\theta})$ forming the
estimate of the target is bounded, then the whole procedure can proceed
without difficulty, assuming the bound provides efficiency of sampling.
However, in the more general situation where the function is not
bounded, there is a complication here in that the unbiased estimate
$\pi({\bolds\theta},{\mathbf u}| {\mathbf y})$ is not guaranteed to be
positive (although its expectation is nonnegative). This issue
prevents us from plugging in directly the estimator $\pi({\bolds\theta},{\mathbf u}|{\mathbf y})$ in the pseudo-marginal framework for the
case of unbounded functions. The problem of such unbiased estimators
returning negative valued estimates turns out to be a well-studied
issue in the Quantum Monte Carlo literature; see, for example, \citet
{lin2000noisy}. The problem is known as the Sign Problem,\footnote
{Workshops devoted to the Sign Problem, for example, the International
Workshop on the Sign Problem in QCD and Beyond, are held regularly,
\surl{http://www.physik.\\uni-regensburg.de/sign2012/}.} which in its
most general form is NP-hard (nondeterministic polynomial time hard)
\citep{PhysRevLett94170201} and at present no general and practical
solution is available. Indeed, recent work by \citet{jacob2013non}
showed that given unbiased estimators of $\lambda\in\mathbb{R}$, no
algorithm exists to yield an unbiased estimate of $f(\lambda)\in
\mathbb{R}^{+}$, where $f$ is a nonconstant real-valued function.
Therefore, we will need to apply a different approach to this problem.

We follow \citet{lin2000noisy} and show that with a weighting of
expectations it is still possible to compute any integral of the form
$\int h({\bolds\theta})\pi({\bolds\theta}|{\mathbf
y})\,d{\bolds\theta}$ by Markov chain Monte Carlo.

\begin{sloppypar}
Suppose that we have an unbiased, but not necessarily positive,
estimate of the likelihood $\hat{p}({\mathbf y}|{\bolds\theta},{\mathbf
u})$ and we wish to sample from $\pi({\bolds\theta},{\mathbf
u}|{\mathbf y}) =\break \hat{p}({\mathbf y}|{\bolds\theta},{\mathbf u})\pi
({\bolds\theta})p({\mathbf u})/p({\mathbf y})$, where $p({\mathbf y}) =
\int\!\!\int p({\mathbf y}|\break {\bolds\theta},{\mathbf u})\pi({\bolds\theta
})p({\mathbf u}) \,d{\bolds\theta} \,d{\mathbf u}$ is an intractable
normaliser. Although $\pi({\bolds\theta},\mathbf u|\mathbf y)$
integrates to one, it is not a probability, as it is not necessarily
positive. Define $\sigma({\mathbf y}|\bolds\theta,{\mathbf u}):=\textsf
{sign}(\hat{p}({\mathbf y}|\bolds\theta,{\mathbf u}))$, where $\textsf
{sign}(x)=1$ when $x>0$, $\textsf{sign}(x)=-1$ if $x<0$ and $\textsf
{sign}(x)=0$ if $x=0$. Furthermore, denote $|\hat{p}(\mathbf y|\bolds\theta,\mathbf u)|$ as the absolute value of the measure, then we have
$\hat{p}({\mathbf y}|\bolds\theta,{\mathbf u})= \sigma({\mathbf
y}|\bolds\theta,{\mathbf u}) |\hat{p}(\mathbf y|\bolds\theta,\mathbf u)|$.
\end{sloppypar}

Suppose that we wish to compute the expectation
%
%
\begin{eqnarray}
\label{eqexpect} \qquad\quad\int h({\bolds\theta})\pi({\bolds\theta
}|{\mathbf y})\,d{
\bolds\theta}=\int\!\!\int h({\bolds\theta}) \pi({\bolds\theta},{\mathbf
u}| {\mathbf y}) \,d {\mathbf u} \,d {\bolds\theta}.
\end{eqnarray}

We can write the above integral as
%
%
\begin{eqnarray}
\label{eqacceptabs}
&& \int h({\bolds\theta})\pi({\bolds\theta
}|{\mathbf y})\,d{\bolds\theta}\nonumber
\\
&&\quad = \int\!\!\int h({\bolds\theta}) \pi(\bolds\theta,{\mathbf
u} | {\mathbf y}) \,d\mathbf{u} \,d {\bolds\theta}
\nonumber
\\
&&\quad = \frac{1}{p({\mathbf y})}\int\!\!\int h({\bolds\theta}) \hat{p}({\mathbf y}|
\bolds\theta,{\mathbf u}) \pi(\bolds\theta)p({\mathbf u}) \,d
{\mathbf u} \,d{
\bolds\theta}
\\
&&\quad = \frac{\int\!\!\int h({\bolds\theta})\sigma({\mathbf
y}|\bolds\theta,{\mathbf u}) |\hat{p}({\mathbf y}|\bolds\theta
,{\mathbf u})| \pi(\bolds\theta)p({\mathbf u}) \,d{\mathbf u}
\,d{\bolds\theta}}{\int\!\!\int\sigma({\mathbf y}|\bolds\theta,{\mathbf u}) |\hat{p}({\mathbf y}|\bolds\theta,{\mathbf u})| \pi
(\bolds\theta)p({\mathbf u}) \,d{\mathbf u} \,d{\bolds\theta}}\hspace*{-20pt}
\nonumber
\\
&&\quad =\frac{\int\!\!\int h({\bolds\theta})\sigma({\mathbf
y}|\bolds\theta,{\mathbf u}) \check\pi(\bolds\theta,{\mathbf u}|{\mathbf y})
\,d{\mathbf u} \,d{\bolds\theta}}{\int\!\!\int\sigma
({\mathbf y}|\bolds\theta,{\mathbf u}) \check\pi(\bolds\theta
,{\mathbf u}|{\mathbf y}) \,d{\mathbf u} \,d{\bolds\theta}},\nonumber
\end{eqnarray}
where $\check\pi(\bolds\theta,\mathbf{u}|\mathbf y)$ is the distribution
\begin{eqnarray*}
\check\pi(\bolds\theta,{\mathbf u}|{\mathbf y}):= \frac{|\hat
{p}({\mathbf y}|\bolds\theta,{\mathbf u})|\pi(\bolds\theta
)p({\mathbf u}) }{\int\!\!\int|\hat{p}({\mathbf y}|\bolds\theta,{\mathbf
u})|\pi(\bolds\theta)p({\mathbf u}) \,d{\mathbf u} \,d{\bolds\theta}}.
\end{eqnarray*}

We can sample from $\check\pi(\bolds\theta,{\mathbf u}|{\mathbf y}) $
using a pseudo-marginal scheme. At each iteration we propose a new
value $\bolds\theta'$, generate an unbiased estimate of the
likelihood $p({\mathbf y}|\bolds\theta',{\mathbf u}')$, and accept it
with probability
\begin{eqnarray*}
\min\biggl\{1, \frac{|\hat{p}(\mathbf y|\bolds\theta',\mathbf{u}')|\pi
(\bolds\theta
')}{|\hat{p}(\mathbf y|\bolds\theta,\mathbf{u})|\pi(\bolds\theta)} \times
\frac{q({\bolds\theta} | {\bolds\theta
}')}{q({\bolds\theta}' | {\bolds\theta})} \biggr\},
\end{eqnarray*}
remembering to save the sign of the accepted estimate. We can then use
Monte Carlo to estimate the expectation in \eqref{eqexpect} using
\eqref{eqacceptabs} with
%
%
\begin{equation}
\label{eqMCsign} \int h({\bolds\theta})\pi({\bolds\theta
}|{\mathbf y})\,d{
\bolds\theta}= \frac{\sum_{i = 1}^{N} h({\bolds\theta_i})\sigma(\mathbf y|\bolds\theta_i,\mathbf{u}_i) }{\sum_{i = 1}^{N}
\sigma(\mathbf y|\bolds\theta_i,\mathbf{u}_i) }.\hspace*{-20pt}
\end{equation}

The output of this MCMC procedure gives an importance-sampling-type
estimate for the desired expectation $\int h({\bolds\theta})\pi
({\bolds\theta}|{\mathbf y})\,d{\bolds\theta}$, which is
consistent but biased (as with estimates from all MCMC methods).
\mbox{Importantly}, this methodology gives us freedom to use unbiased
estimators which may occasionally return negative estimates. We
describe the procedure more systematically in the \hyperref[app]{Appendix}
(Section~\ref{secabsmeas}), and we discuss in particular how to
compute the effective sample size of the resulting Monte Carlo estimate.

The following section addresses the issue of constructing the unbiased
estimator to be used in the overall MCMC scheme.

\section{Pseudo-Marginal MCMC for Doubly-Intractable
Distributions}\label{secPMdoub}

The foundational component of pseudo-marginal MCMC is the unbiased and
positive estimator of the target density. In the methodology developed
here, it is not essential for the estimate of the intractable
distribution to be strictly positive and we exploit this
characteristic. Note that whilst there are many methods for unbiasedly
estimating ${\mathcal Z}(\bolds\theta)$, such as importance
sampling, Sequential Monte Carlo (SMC) \citep{delmoral2006} and
Annealed Importance Sampling (AIS) \citep{Neal98annealedimportance},
if we then take some nonlinear function of the estimate, for example,
the reciprocal, the overall estimate of the likelihood is no longer unbiased.

It is possible to directly construct an estimator of $1/{\mathcal
{Z}(\bolds{\theta})}$ using an instrumental density $q({\mathbf y})$ as follows:
\begin{eqnarray*}
\frac{1}{\mathcal{Z}({\bolds\theta})} &=& \frac{1}{\mathcal
{Z}({\bolds\theta})} \int q({\mathbf y}) \,d {\mathbf y} =
\int\frac{q({\mathbf y})}{f({\mathbf y};\theta)} p({\mathbf y}|{\bolds\theta
})\,d {\mathbf y}
\\
&\approx&
\frac{1}{N} \sum_{i = 1}^{N}
\frac
{q({\mathbf y}_i)}{f({\mathbf y}_i;\theta)}, \quad{\mathbf y}_i \sim p(\cdot
|{\bolds\theta});
\end{eqnarray*}
however, this requires the ability to sample from the likelihood, and
if we can do this, then we can implement the Exchange algorithm.
Further, the variance of the estimate depends strongly on the choice of
the instrumental density. A biased estimator can be constructed by
sampling the likelihood using MCMC (e.g., \cite{zhang2012continuous}), but a pseudo-marginal scheme based on
this estimate will not target the correct posterior distribution. Very
few methods to estimate $1/{\mathcal{Z}({\bolds\theta})}$ can
be found in the Statistics or Physics literature, presumably because in
most situations a consistent estimate will suffice. Therefore, we have
to look for other ways to generate an unbiased estimate of the likelihood.

In outline, the intractable distribution is first written in terms of a
nonlinear function of the nonanalytic/computable normalising term. For
example, in equation (\ref{eqnposterior}), the nonlinear function is
the reciprocal $1/{\mathcal Z}({\bolds\theta})$, and an
equivalent representation would be $\exp(-\log{\mathcal
Z}({\bolds\theta}))$. This function is then represented by a
convergent Maclaurin expansion which has the property that each term
can be estimated unbiasedly using the available unbiased estimates of
$\hat{\mathcal Z}({\bolds\theta})$. The infinite series
expansion is then stochastically truncated without introducing bias so
that only a finite number of terms need be computed. These two
components---(1) unbiased independent estimates of the normalising
constant, and (2) unbiased stochastic truncation of the infinite series
representation---then produce an unbiased, though not strictly
positive, estimate of the intractable distribution. The final two
components of the overall methodology consist of (3) constructing an
MCMC scheme which targets a distribution proportional to the absolute
value of the unbiased estimator, and then (4) computing Monte Carlo
estimates with respect to the desired posterior distribution as
detailed in the previous section.

This method has its roots in several places in the Statistics and
Physics literature. In the Physics literature, researchers used a
similar method to obtain unbiased estimates of $\exp(-U(x))$ when only
unbiased estimates of $U(x)$ were available (\cite{kennedy1985noise}, \cite{bhanot1985bosonic}). They further showed that even
when using such unbiased estimates in place of the true value, detailed
balance still held. The method for realising the unbiased estimates at
each iteration is also similar to that suggested by \citet
{booth2007unbiased}, in which he described a method for unbiasedly
estimating the reciprocal of an integral, which is of obvious relevance
to our case. In the Statistics literature, \citet{douc2011vanilla}
used a geometric series to estimate an inverse probability, and
\citet{Beskos2006id}, \citet{Fearnhead2008kk} also used techniques to truncate a
series unbiasedly in their work on likelihood estimation for stochastic
diffusions. Finally, both \citet{rhee2012new} and \citet
{mcleish2011general} use roulette methods to realise an unbiased
estimate when only biased but consistent estimates are available. This
is achieved by writing the quantity to be unbiasedly estimated as an
infinite series in which each term is a function of the consistent
estimates which can be generated, and then truncating the series using
roulette methods.

In the following sections, we study two series expansions of a
doubly-intractable likelihood, in which each term can be estimates
unbiasedly using unbiased estimates of ${\mathcal Z}({\bolds\theta})$. Following this comes a description of unbiased truncation methods.

\subsection{Geometric Series Estimator}\label{secgeo}

In the following discussion we show how the intractable likelihood can
be written as a geometric series in which each term can be estimated
unbiasedly. Take a biased estimate of the likelihood $\tilde{p}({\mathbf
y} |{\bolds\theta}) = f({\mathbf y}; {\bolds\theta})/\widetilde
{\mathcal Z}({\bolds\theta})$, where $\widetilde{\mathcal
Z}({\bolds\theta})>0$ is ideally an upper bound on ${\mathcal
Z}(\bolds{\theta})$ or, alternatively, an unbiased importance sampling
estimate or a deterministic approximation. Then, using a multiplicative
correction
%
%
\begin{equation}
p({\mathbf y}|\bolds\theta) = \tilde{p}({\mathbf y}|\bolds\theta)
\times c({
\bolds\theta}) \Biggl[1 + \sum_{n=1}^{\infty}
\kappa({\bolds\theta})^n \Biggr],
\end{equation}
where $\kappa({\bolds\theta}) = 1 - c({\bolds\theta
}){{\mathcal Z}({\bolds\theta})}/{\widetilde{\mathcal
Z}({\bolds\theta})}$ and $c({\bolds\theta})$ ensures
$|\kappa({\bolds\theta})| < 1$, the convergence of a geometric
series gives
\begin{eqnarray*}
\tilde{p}({\mathbf y}|{\bolds\theta}) \times c({\bolds\theta})
\Biggl[1 +
\sum_{n=1}^{\infty}\kappa({\bolds\theta
})^n \Biggr]
&=& \tilde{p}({\mathbf y}|{\bolds\theta}) \times
\frac{c({\bolds\theta})}{1 - \kappa({\bolds\theta})}
\\
&=& \tilde{p}({\mathbf
y}|{\bolds\theta}) \times
\frac{\widetilde{\mathcal Z}({\bolds\theta})}{{\mathcal Z}({\bolds\theta})}
\\
&=&
p({\mathbf y}|{\bolds\theta}).
\end{eqnarray*}

Based on this equality, and with an infinite number of independent
unbiased estimates of ${\mathcal Z}({\bolds\theta})$ each
denoted $\hat{\mathcal Z}_i({\bolds\theta})$, an unbiased
estimate of the target density is
%
%
\begin{eqnarray}\label{geomest}
\hat\pi({\bolds\theta}|{\mathbf y}) &=& \frac{\pi
(\bolds\theta)\tilde{p}(\mathbf y|\bolds\theta)}{p(\mathbf y)}
\nonumber\\[-8pt]\\[-8pt]\nonumber
&&{}\cdot
c({\bolds\theta}) \Biggl[1 + \sum_{n=1}^{\infty}
\prod_{i=1}^n \biggl(1 - c({\bolds\theta})\frac{\hat{\mathcal
Z}_i({\bolds\theta})}{\widetilde{\mathcal Z}({\bolds\theta})} \biggr) \Biggr].\hspace*{-20pt}
\end{eqnarray}

Notice that the series in (\ref{geomest}) is finite a.s. and we can
interchange summation and expectation if
\begin{eqnarray*}
E \biggl(\biggl\llvert1 - c({\bolds\theta})\frac{\hat{\mathcal
Z}_i({\bolds\theta})}{\widetilde{\mathcal Z}({\bolds\theta})}\biggr\rrvert
\biggr)<1.
\end{eqnarray*}

Since $E(|X|)\leq E^{1/2}(|X|^2)$, a sufficient condition for this is
$0<c(\bolds\theta)<2\widetilde{\mathcal Z}({\bolds\theta
}) {\mathcal Z}({\bolds\theta})/ E (\hat{\mathcal
Z}^2_1({\bolds\theta}) )$, which is slightly more
stringent than $|\kappa(\bolds{\theta})|<1$. Under this assumption, the
expectation of $\hat\pi({\bolds\theta}|{\mathbf y})$ is
\begin{eqnarray*}
\nonumber
&& E \bigl\{\hat\pi({\bolds\theta}|{\mathbf y})|\widetilde{\mathcal Z}({
\bolds\theta}) \bigr\}
\\
&&\quad = \frac{\pi(\bolds\theta
) \tilde{p}(\mathbf y|\bolds\theta)}{p(\mathbf y)}\nonumber
\\
&&\qquad{}\cdot c({\bolds\theta})
\Biggl[1 + \sum_{n=1}^{\infty}\prod
_{i=1}^n \biggl(1 - c({\bolds\theta})
\frac{E \{\hat{\mathcal
Z}_i({\bolds\theta}) \}}{\widetilde{\mathcal
Z}({\bolds\theta})} \biggr) \Biggr]
\\
&&\quad = \frac{\pi(\bolds\theta) \tilde{p}({\mathbf y}|\bolds\theta
)}{p(\mathbf y)}\times c({\bolds\theta}) \Biggl[1 + \sum
_{n=1}^{\infty}\kappa({\bolds\theta})^n \Biggr]\nonumber
\\
\nonumber
&&\quad  = \pi({\bolds\theta}|{\mathbf y}).
\end{eqnarray*}

Therefore, the essential property $E \{\hat\pi({\bolds\theta}|{\mathbf y}) \} = \pi({\bolds\theta}|{\mathbf y})$
required for Exact-Approximate MCMC is satisfied by this geometric
correction. However, there are difficulties with this estimator. It
will be difficult in practice to find $c(\bolds\theta)$ that
ensures the series in (\ref{geomest}) is convergent in the absence of
knowledge of the actual value of ${\mathcal Z}({\bolds\theta})$.
By ensuring that $\widetilde{\mathcal Z}({\bolds\theta
})/c({\bolds\theta})$ is a strict upper bound on ${\mathcal
Z}({\bolds\theta})$, denoted by ${\mathcal Z}_U$, guaranteed
convergence of the geometric series is established. Even if an upper
bound is available, it may not be computationally practical, as upper
bounds on normalising constants are typically loose (see, e.g., \cite{el2008convex}), making the ratio ${\mathcal Z}({\bolds\theta})/{\mathcal Z}_U$ extremely small, and, therefore, $\kappa
({\bolds\theta})\approx1$; in this case, the convergence of the
geometric series will be slow. A more pragmatic approach is to use a
pilot run at the start of each iteration to characterise the location
and variance of the ${\mathcal Z}({\bolds\theta})$ estimates,
and use this to conservatively select $\widetilde{\mathcal
Z}({\bolds\theta})/c({\bolds\theta})$ such that the series converges. Of
course, if the distribution of the estimates is not well enough
characterised, then we may not be able to guarantee with probability 1
that $|\kappa({\bolds\theta})|<1$, and hence approximation will
be introduced into the chain.

In the next section we describe an alternative to the geometric series
estimator which does not have the practical issue of ensuring the
region of convergence is maintained.

\subsection{Unbiased Estimators Using an Exponential Auxilliary
Variable}\label{secmaclaurin}

In this section we show how the introduction of an auxiliary variable
can enable the posterior density to be written in terms of a Taylor
series expansion of the exponential function. The introduction of $\nu
\sim\textsf{Expon}({\mathcal Z}({\bolds\theta}))$ defines a
joint distribution of the form of
\begin{eqnarray*}
\pi({\bolds\theta}, \nu|{\mathbf y})&=& \bigl[{\mathcal Z}({\bolds\theta})
\exp\bigl(-\nu{\mathcal Z}({\bolds\theta})\bigr) \bigr]
\\
&&{}\cdot
\frac{f({\mathbf y};{\bolds\theta})}{{\mathcal Z}({\bolds\theta})} \times\pi({\bolds\theta}) \times\frac{1}{p({\mathbf y})}
\\
&=& \exp\bigl( -\nu{\mathcal Z}({\bolds\theta}) \bigr) \times
f({\mathbf y};{
\bolds\theta}) \times\pi({\bolds\theta}) \times\frac
{1}{p({\mathbf y})}
\\
&=& \Biggl[ 1 + \sum_{n = 1}^{\infty}
\frac{(-\nu{\mathcal
Z}({\bolds\theta}) )^n}{n!} \Biggr]
\\
&&{}\cdot f({\mathbf y};{\bolds\theta}) \times\pi({
\bolds\theta}) \times\frac{1}{p({\mathbf y})}.
\end{eqnarray*}

Integrating over $\nu$ returns the posterior distribution and,
therefore, if we sample from this joint distribution, our $\bolds\theta$ samples will be distributed according to the posterior. As
hinted at in the previous section, the methods used to truncate the
series are more computationally feasible if the series converges
quickly. Therefore, we introduce $\widetilde{{\mathcal Z}}({\bolds\theta})$, which is preferably an upper bound on ${\mathcal
Z}({\bolds\theta})$ or, if unavailable, some other
approximation. The exponential can then be expanded as follows:
\begin{eqnarray*}
\exp\bigl(-\nu{{\mathcal Z}}({\bolds\theta})\bigr) &=& \exp\bigl
(-\nu\widetilde
{{\mathcal Z}}({\bolds\theta})\bigr)
\\
&&{}\cdot \exp\bigl(\nu\bigl
(\widetilde{{
\mathcal Z}}({\bolds\theta}) - {\mathcal Z}({\bolds\theta
})\bigr)
\bigr)
\\
&=& \exp\bigl(-\nu\widetilde{{\mathcal Z}}({\bolds\theta})\bigr)
\\
&&{}\cdot
\Biggl(1 + \sum_{n=1}^{\infty} \frac{\nu^n}{n!}
\bigl( \widetilde{{\mathcal Z}}({\bolds\theta}) -{\mathcal
Z}({\bolds\theta}) \bigr)^n \Biggr).
\end{eqnarray*}

If $\widetilde{\mathcal Z}({\bolds\theta})$ is an upper bound on ${\mathcal
Z}(\bolds\theta)$, then its introduction prevents the terms in
the Taylor series from alternating in sign by ensuring the exponent is
positive; this helps to reduce the impact of returning negative
estimates. Even if $\widetilde{\mathcal Z}({\bolds\theta})$ is not a strict
upper bound, its presence reduces the absolute value of the exponent,
which improves the convergence properties of the series, and therefore
makes the truncation methods described in the next section more efficient.

An unbiased estimator of the series is
%
%
\begin{eqnarray}\label{poissonestimator}
&& \widehat{\exp\bigl(-\nu{\mathcal
Z}({\bolds\theta})\bigr)}\nonumber
\\
&&\quad  =
\exp\bigl(-\nu\widetilde{{\mathcal Z}}({\bolds\theta})\bigr)
\\
&&\qquad{}\cdot \Biggl[1 + \sum
_{n=1}^{\infty} \frac{\nu^n}{n!}\prod
_{i=1}^n \bigl({\widetilde{\mathcal Z}}({\bolds\theta}) - \hat{\mathcal Z}_i({\bolds\theta}) \bigr) \Biggr],\nonumber
\end{eqnarray}
where $\{\hat{\mathcal Z}_i({\bolds\theta}), i\geq1\}$ are i.i.d.
random variables with expectation equal to ${\mathcal Z}({\bolds\theta
})$. The magnitude of the exponent can present computational barriers
to the implementation of this scheme; if ${\mathcal Z}({\bolds\theta})$
is very large, it is easier to carry out the division $\hat{\mathcal
Z}({\bolds\theta}) / {\mathcal Z}({\bolds\theta})$ in \eqref{geomest}
(which can be computed in log space) than the subtraction ${\mathcal
Z}({\bolds\theta}) - \hat{\mathcal Z}({\bolds\theta}) $ in \eqref
{poissonestimator}. On the other hand, since $n!$ grows faster than the
exponential, this series is always well defined (finite almost surely).

In Fearnhead, Papaspiliopoulos and Roberts\break (\citeyear{Fearnhead2008kk}), the {\em Generalised Poisson Estimator},
originally proposed in \citet{Beskos2006id}, is employed to estimate
transition functions that are similar to \eqref{poissonestimator}.
Here again, this series is finite almost surely with finite
expectation. The choice of which estimator to employ will be problem
dependent and, in situations where it is difficult to guarantee
convergence of the geometric series, this form of estimator may be more
suitable.

In the following section, we discuss the final element of the proposed
methodology: unbiased truncation of the infinite series estimators.

\section{Unbiased Truncation of Infinite Sums: Russian Roulette}\label
{secunbiased}

Two unbiased estimators of nonlinear functions of a normalising
constant have been considered. Both of them rely on the availability of
an unbiased estimator for ${\mathcal Z}(\bolds\theta)$ and a
series representation of the nonlinear function. We now require a
computationally feasible means of obtaining the desired estimator
without explicitly computing the infinite sum and without introducing
any bias into the final estimate. It transpires that there are a number
of ways to randomly truncate the convergent infinite sum ${\mathcal
S}({\bolds\theta}) = \sum\nolimits_{i=0}^{\infty} \phi
_i({\bolds\theta})$ in an unbiased manner. These stem from work
by von Neumann and Ulam in the 1940s; see \citet{papaspil2011} for a
good review of such methods.

\subsection{Single Term Weighted Truncation}
The simplest unbiased truncation method is to define a set of
probabilities and draw an integer index $k$ with probability $q_k$,
then return $\phi_k(\theta)/q_k$ as the estimator. It is easy to see
that the estimator is unbiased as $E\{\hat{\mathcal S}(\bolds\theta)\} = \sum\nolimits_k q_k \phi_k({\bolds\theta})/q_k =
{\mathcal S}(\bolds\theta)$. The definition of the probabilities
should be chosen to minimise the variance of the estimator; see, for
example, \citet{Fearnhead2008kk}. An example could be that each index
is drawn from a Poisson distribution $k \sim\textsf{Poiss}(\lambda)$
with $q_k = \lambda^k\exp(-\lambda)/k!$. However, in the case of a
geometric series where $\phi_k(\bolds\theta)=\phi
^k(\bolds\theta)$, the variance of the estimator will be
infinite with this choice since the combinatorial function $k!$ grows
faster than the exponential. Using the geometric distribution as our
importance distribution, the variance is finite subject to some
conditions on the choice of $p$, the parameter of the geometric
distribution. To see~this,\vspace*{2pt} note that, as $k$ is chosen with probability
$q_k = p^k(1-p)$, the second moment $\mathbb{E}[\hat{S}^2] = \sum_{k
= 0}^{\infty} \hat{S}_k^2 q_k = \sum_{k = 0}^{\infty} \phi_k^2
/p^k (1-p)$ is finite if $\lim_{k\to\infty} |\phi_{k+1}^{2}/p\phi
_{k}^{2} |<1$.

\subsection{Russian Roulette}\label{subsecRR}

An alternative unbiased truncation that exhibits superior performance
in practice is based on a classic Monte Carlo scheme, known as Russian Roulette in the Physics literature (\cite{lux1991monte}, \cite{carter1975particle}). The procedure is based on the
simulation of a finite random variable (stopping time) $\tau$
according to some probabilities $p_n = \mathbb{P}(\tau\geq n)>0$ for
all $n\geq0$ with $p_0 = 1$. Define the weighted partial sums as $S_0
= \phi_0$ and for $k\geq1$
\begin{eqnarray*}
S_k = \phi_0 + \sum_{j = 1}^{k}
\frac{\phi_j}{p_j}.
\end{eqnarray*}

The Russian Roulette estimate of $S$ is $\hat{S} = S_{\tau}$. Russian Roulette implementations in the Physics literature commonly choose a
stopping time of the form
\begin{eqnarray*}
\tau=\inf\{k\geq1: U_k\geq q_k \},
\end{eqnarray*}
where $\{U_j, j\geq1\}$ are i.i.d. $\mathcal{U}(0,1)$, $q_j\in
(0,1]$ and $\hat{S} = S_{\tau-1}$. In this case $p_n=\prod_{j=1}^{n-1} q_j$.

It can be shown that the expectation of the estimate is as required:
\begin{eqnarray*}
\sum_{k=0}^n S_k\PP(\tau=k) &=&
\sum_{k=0}^{n} S_k
(p_k - p_{k+1})
\\
&=& \phi_0 + \sum
_{k=0}^{n-1}S_{k+1}p_{k+1} - \sum
_{k = 0}^{n} S_k p_{k+1}
\\
&=&\sum_{k=0}^n\phi_k
-S_np_{n+1}.
\end{eqnarray*}

By Kronecker's lemma, $\lim_{n\to\infty} p_n S_n=0$, and $|p_{n+1}
S_n|=(p_{n+1}/p_n)p_n|S_n|\leq p_n|S_n|\to0$, as\vspace*{1pt} $n\to\infty$. We
conclude that $\PE[\hat S(\bolds\theta)]=\sum_{k=0}^\infty
S_k\PP(\tau=k) = \sum_{k=0}^\infty\phi_k = S(\bolds\theta
)$. We refer the reader to the \hyperref[app]{Appendix} for a more detailed discussion
relating to the variance of such an estimator and how to design the
sequence of probabilities $(p_n)$.

Based on results presented in the \hyperref[app]{Appendix}, for a geometric series
where $\phi_k(\bolds\theta)=\phi^k(\bolds\theta)$, if
one chooses $q_j=q$, then the variance will be finite provided $q> \phi
(\bolds\theta)^2$. In general, there is a trade-off between the
computing time of the scheme and the variance of the returned estimate.
If the selected $q_j$'s are close to unity, the variance is small, but
the computing time is high. But if $q_j$'s are close to zero, the
computing time is fast, but the variance can be very high, possibly
infinite. In the case of the geometric series, $\phi_k(\bolds\theta)=\phi^k(\bolds\theta)$, choosing $q_j=q=\phi
(\bolds\theta)$ works reasonably well in practice.

As an illustrative example, consider the joint density
%
%
\begin{eqnarray}
p({\bolds\theta}, \nu, {\mathbf u}|{\mathbf y}) &=& \exp\bigl(-\nu\widetilde
{{\mathcal
Z}}({\bolds\theta})\bigr) \nonumber
\\
&&{}\cdot \Biggl(1 + \sum
_{n=1}^{\tau_{\theta}} \frac{\nu^n}{q^n n!}\prod
_{i=1}^n \bigl( \widetilde{{\mathcal Z}}({\bolds\theta}) - \hat{{\mathcal Z}_i}({\bolds\theta}) \bigr) \Biggr)
\\
&&{}\cdot \frac
{f({\mathbf y};{\bolds\theta}) \pi({\bolds\theta}) } {p({\mathbf y})},\nonumber
\end{eqnarray}
where the random variable ${\mathbf u}$ represents the random variables in
the estimates $\hat{{\mathcal Z}_i}({\bolds\theta})$ and the
random variable used in Russian Roulette truncation, and $q^n = \prod
_{l=1}^n q_l$ denotes the probabilities in the Russian Roulette
truncation. If we define a proposal for $\nu'$ as $q(\nu
'|{\bolds\theta}') = \widetilde{\mathcal Z}({\bolds\theta}')
\exp(-\nu' \widetilde{\mathcal Z}({\bolds\theta}'))$ and a
proposal for ${\bolds\theta}'$ as $q({\bolds\theta
}'|{\bolds\theta})$, then the Hastings ratio for a transition
kernel with invariant density $\pi({\bolds\theta}, \nu, {\mathbf
u}|{\mathbf y})$ follows as
%
%
\begin{eqnarray}
&& \frac{ f({\mathbf y};{\bolds\theta}') }{ f({\mathbf y};{\bolds\theta}) } \times\frac{\widetilde{\mathcal Z}({\bolds\theta
})}{\widetilde{\mathcal
Z}({\bolds\theta}')}\times \frac{ \pi({\bolds\theta}')}{ \pi
({\bolds\theta})}
\nonumber\\[-8pt]\\[-8pt]\nonumber
&&\quad{}\cdot
\frac{q({\bolds\theta}|{\bolds\theta}')}{q({\bolds\theta}'|{\bolds\theta})}\times\phi\bigl(\nu, \nu', {\bolds\theta}, {
\bolds\theta}'\bigr),
\end{eqnarray}
where
%
%
\begin{eqnarray}
&& \phi\bigl(\nu, \nu', {\bolds\theta}, {\bolds\theta}'\bigr)
\nonumber\\[-8pt]\\[-8pt]\nonumber
&&\quad = \frac
{1 + \sum_{m=1}^{\tau_{\theta'}} \frac{(\nu')^m}{q^m m!}\prod_{j=1}^m (
\widetilde{{\mathcal Z}}({\bolds\theta}') - \hat
{{\mathcal Z}_j}({\bolds\theta}') )}{
1 + \sum_{n=1}^{\tau_{\theta}} \frac{\nu^n}{q^n n!}\prod_{i=1}^n (
\widetilde{{\mathcal Z}}({\bolds\theta}) - \hat
{{\mathcal Z}_i}({\bolds\theta}) ) }.
\end{eqnarray}

It is interesting to note that $\phi(\nu, \nu', {\bolds\theta
}, {\bolds\theta}') $ acts as a multiplicative\vspace*{1pt} correction for
the Hastings ratio that uses the approximate normalising term
$\widetilde
{\mathcal Z}({\bolds\theta})$ rather than the actual ${\mathcal
Z}({\bolds\theta})$. The required marginal $\pi({\bolds\theta}|{\mathbf y})$ follows due to the unbiased nature of the estimator.

The Russian Roulette methodology has been used in various places in the
literature. \citet{mcleish2011general} and \citet{rhee2012new},
\citet{glynn2014exact} cleverly use the Russian Roulette
estimator to ``debias'' a biased but consistent estimator. We would like
to unbiasedly estimate $X$, for which we have available only a \mbox{sequence}
of approximations, $X_i$, with $\mathrm{E}[X_i] \to\mathrm{E}[X]$ as
$i\to\infty$. Define an infinite series, $S = \mathrm{X_0} + \sum_{n =
1}^{\infty} (X_n - X_{n-1})$; an unbiased estimate of $S$ is an
unbiased estimate of $X$, assuming that the estimates are good enough
to interchange expectation and summation. To achieve a computationally
feasible and unbiased estimator of $X$, the Roulette or Poisson
truncation schemes can then be applied. In the context of our work,
this provides an alternative to the geometric or exponential series
described above, in which only a consistent estimator is required. One
drawback to this debiasing scheme for use in pseudo-marginal MCMC is
that there is no obvious way to reduce the probability of the final
estimate being negative. Russian Roulette is also employed extensively
in the modelling of Neutron Scattering in Nuclear Physics and Ray
Tracing in Computer Graphics (\cite{hendricks1985mcnp}, \cite{carter1975particle}).

Now that the complete Exact-Approximate\break MCMC scheme has been detailed,
the following section illustrates the methodology on some models that
are doubly-intractable, considering the strengths and weaknesses.

\section{Experimental Evaluation}\label{secexp}

\subsection{Ising Lattice Spin Models}

Ising models are examples of doubly-intractable distributions over
which it is challenging to perform inference. They form a prototype for
priors for \mbox{image} segmentation and autologistic models, for example,
\citet{hughes2011autologistic},
\citet{gu2001maximum},
\citet{moller2006efficient}.
Current exact methods such as the Exchange algorithm \citep
{murray2006mcmc} require access to a perfect sampler \citep
{propp1996exact}, which, while feasible for small grids, cannot be
scaled up. A practical alternative is employed in \citet
{caimo2011bayesian}, where an auxiliary MCMC run is used to
approximately simulate from the model. This is inexact and introduces
bias, but it is hoped that the bias has little practical impact. We
compare this approximate scheme with our exact methodology in this section.

For an $N \times N$ grid of spins, $\mathbf{y} = (y_1, \dots,y_{N^2})$, $y\in\{+1,-1\}$, the Ising model has likelihood
%
%
\begin{eqnarray}
&& p(\mathbf{y};\alpha,\beta)
\nonumber\\[-8pt]\\[-8pt]\nonumber
&&\quad = \frac{1}{{\mathcal Z}(\alpha,\beta
)}\exp\Biggl( \alpha\sum
_{i}^{N^2} y_i + \beta\sum
_{i \sim j} y_i y_j \Biggr),
\end{eqnarray}
where $i$ and $j$ index the rows and column of the lattice and the
notation $i \sim j$ denotes summation over nearest neighbours. Periodic
boundary conditions are used in all subsequent computation. The
parameters $\alpha$ and $\beta$ indicate the strength of the external
field and the interactions between neighbours, respectively. The
normalising constant,
%
%
\begin{equation}
\qquad {\mathcal Z}(\alpha,\beta) = \sum_{{\mathcal Y}} \exp\Biggl(
\alpha\sum_{i}^{N^2} y_i +
\beta\sum_{i \sim j} y_i y_j
\Biggr),
\end{equation}
requires summation over all $2^{N^2}$ possible configurations of the
model, which is computationally infeasible even for moderately sized
lattices. This is, in fact, a naive bound as the transfer matrix method
(see, e.g., \cite{mackay2003information}), which has complexity
$N 2^{N}$ that can also be used to compute the partition function.

%
\begin{table*}[t]
\tabcolsep=0pt
\caption{Monte Carlo estimates of the mean and standard deviation of
the posterior distribution $p(\beta|\mathbf{y})$ using the five
algorithms described. The debiasing series estimates have been
corrected for negative estimates. The exact chain was run for 100,000
iterations and then the second half of samples used to achieve a ``gold
standard'' estimate. An estimate of the effective sample size (ESS) is
also shown based on 10,000 MCMC samples}\label{IsingTable}
\begin{tabular*}{\tablewidth}{@{\extracolsep{\fill}}@{}lcc c c c@{}}
\hline
 & \textbf{Roulette} & \textbf{Poisson} & \textbf{Exchange (approx)} & \textbf{Exchange (exact)} & \textbf{Exact}\\
\hline
Mean & 0.2004 & 0.2005  &  0.2013  &  0.2010  & 0.2008 \\
Standard deviation &  0.0625 &  0.0626  &  0.0626  & 0.0626  & 0.0625 \\
ESS &  2538  &  2660  &  1727  &  1732  &  3058  \\
\hline
\end{tabular*}
\end{table*}

%
\begin{figure*}[b]

\includegraphics{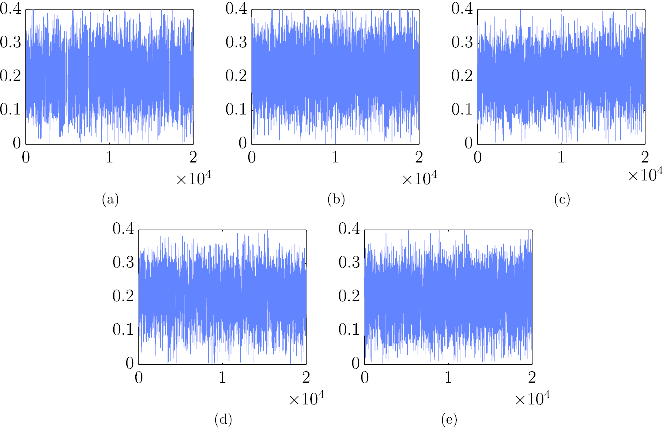}

\caption{Traces of samples using the debiasing infinite series with
\textup{(a)} Russian Roulette,
\textup{(b)} Poisson truncation, and
\textup{(c)} the approximate Exchange algorithm,
\textup{(d)} the Exchange algorithm using perfect samples and
\textup{(e)} an MCMC chain with the partition function calculated using the
matrix transfer method. Note in \textup{(a)} and \textup{(b)} the samples are not drawn
from the posterior distribution, $p(\beta|{\mathbf y})$, but from the
(normalised) absolute value of the estimated density.}\label{figIsingTraces}
\end{figure*}

%
\begin{figure*}[t]

\includegraphics{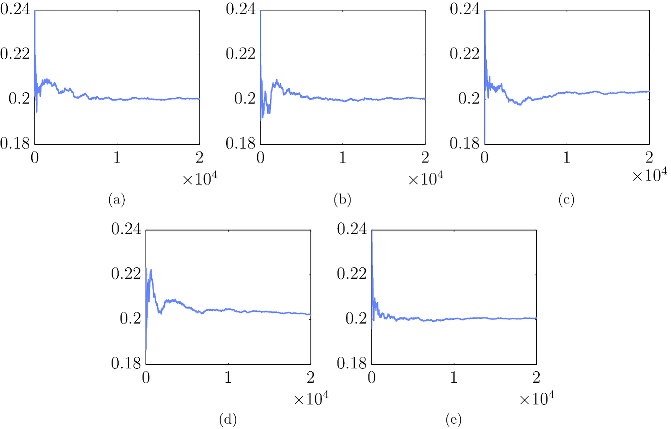}

\caption{Plots of the running mean for the posterior distribution
$p(\beta|\mathbf{y})$ of a $10 \times 10$ Ising model using three
methods: \textup{(a)} debiasing series with roulette truncation,
\textup{(b)} debiasing series with Poisson truncation,
\textup{(c)} approximate Exchange,
\textup{(d)} the Exchange algorithm using perfect samples and
\textup{(e)} an MCMC chain with the
partition function calculated using the matrix transfer method.}\vspace*{-5pt}
\label{figIsingRunMean}
\end{figure*}

Experiments were carried out on a small $10 \times 10$ lattice to
enable a detailed comparison of the various algorithms. A configuration
was simulated using a perfect sampler with parameters set at $\alpha=
0$ and $\beta= 0.2$. Inference was carried out over the posterior
distribution $p(\beta|\mathbf{y})$ ($\alpha= 0$ was fixed). A
standard Metropolis--Hastings sampling scheme was used to sample the
posterior, with a normal proposal distribution centred at the current
value and acceptance rates tuned to around 40\%. A uniform prior on
$[0,1]$ was set over $\beta$. As no tight upper bound is available on
the normalising term ${\mathcal{Z}}(\bolds\theta)$, the
debiasing series construction of \citet{mcleish2011general} and \citet
{glynn2014exact}, described at the end of Section~\ref{subsecRR}, was
used to construct an unbiased estimate of the likelihood. The sequence
of biased but consistent estimates of $1/{\mathcal{Z}}(\bolds\theta)$ was produced by taking the reciprocal of unbiased SMC
estimates of ${\mathcal Z}(\bolds\theta)$ with an increasing
number of importance samples and temperatures [see \citet
{delmoral2006} for a good introduction to SMC]. SMC proceeds by
defining a high-dimensional importance density which is sampled
sequentially, and in this case we used a geometric schedule (\cite{gelman1998simulating},
\cite{Neal98annealedimportance}) to define the sequence
of distributions
\begin{eqnarray*}
p({\mathbf y}|{\bolds\theta})_n \propto p({\mathbf y}|{\bolds\theta})^{\phi_n} U({\mathbf y})^{1 - \phi_n},
\end{eqnarray*}
with $0\leq\phi_1< \cdots<\phi_p = 1$ and $U(\cdot)$ a uniform
distribution over all the grids in $\mathcal{Y}$. A Gibbs transition
kernel, in which one spin was randomly selected and updated according
to its conditional distribution, was used to sequentially sample the
high-dimensional space. The initial estimate, $1/{\mathcal
{Z}}({\bolds\theta})_0$, used 100 temperatures and 100
importance samples; the $i$th estimate used $100\times2^i$
temperatures and importance samples.

The infinite series was truncated unbiasedly using both Poisson
truncation and Russian Roulette. For comparison, the posterior
distribution was also sampled using the Exchange algorithm, the
approximate form of the Exchange algorithm \citep{caimo2011bayesian}
with an auxiliary Gibbs sampler run for 50,000 steps at each iteration,
and an ``exact'' MCMC chain using the matrix transfer method to calculate
the partition function at each iteration. All chains were run for
20,000 iterations and the second half of the samples used for Monte
Carlo estimates.

The exact posterior mean and standard deviation are not available for
comparison, but the estimates from the five methods agree well
(Table~\ref{IsingTable}). The traces in Figure~\ref{figIsingTraces}
show that the algorithms mix well and Figures~\ref{figIsingRunMean}
and \ref{figIsingRunStd} show that the estimates of the mean and
standard deviation agree well. Estimates of the Effective sample size
(ESS) are also included in Table~\ref{IsingTable}, which give an idea
of how many independent samples are obtained from each method per
10,000 samples.

%
\begin{figure*}[b]

\includegraphics{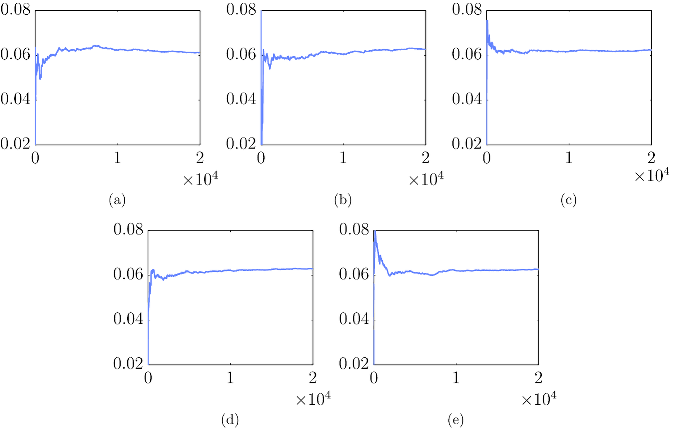}

\caption{Plots of the running standard deviation for the posterior
distribution $p(\beta|\mathbf{y})$ of a $10 \times 10$ Ising model
using three methods: \textup{(a)}~debiasing series with roulette truncation,
\textup{(b)} debiasing series with Poisson truncation,
\textup{(c)} approximate Exchange,
\textup{(d)} the Exchange algorithm using perfect samples and
\textup{(e)} an MCMC chain with
the partition function calculated using the matrix transfer method.}\vspace*{-5pt}\label{figIsingRunStd}
\end{figure*}
%
%
\begin{figure*}[t]

\includegraphics{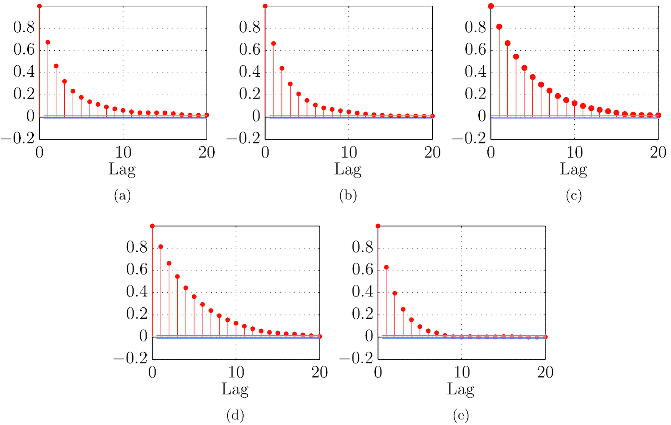}

\caption{Autocorrelation plots for samples drawn from the posterior
distribution $p(\beta|\mathbf{y})$ of a $10 \times 10$ Ising model
using five methods:
\textup{(a)}~debiasing series with roulette truncation,
\textup{(b)}~debiasing series with Poisson truncation,
\textup{(c)}~approximate Exchange,
\textup{(d)}~the Exchange algorithm using perfect samples and
\textup{(e)} an MCMC chain with the partition function calculated using the matrix transfer method.}\label{figIsingACF}
\end{figure*}

Approximately 5\% of estimates were negative when using roulette
truncation and 10\% when using Poisson truncation; however, using the
correction in equation~\eqref{eqMCsign}, expectations with respect to
the posterior still converge to the correct values. If we had opted to
implement the geometric series construction of Section~\ref{secgeo} in
order to reduce the number of negative estimates, we have available
only a naive upper bound for the partition function corresponding to
setting all spins to $+$1. This bound is very loose and therefore
impractical, as the series converges very slowly. Hence, the
availability of a method to deal with negative estimates frees us from
atrocious upper bounds that would explode the asymptotic variance of
the chains.

The autocorrelation functions (Figure~\ref{figIsingACF}) and the
effective sample size (Table~\ref{IsingTable}) of both Russian Roulette and Poisson truncation outperform the approximate and exact
Exchange algorithm in this example and are comparable to the exact
implementation; of course, it is possible to improve the performance of
our algorithm by using more computation, whereas this is not possible
with the Exchange algorithm. It should be noted that the Exchange
algorithm in this guise is less computationally intensive. However, it
becomes impossible to perfectly sample as the size of the lattice
increases, whereas our algorithm can still be implemented, albeit with\vadjust{\goodbreak}
considerable computational expense. Note that even at this small
lattice size, the approximate version of Exchange looks noticeably less stable.

We have further experimented on larger lattices, for example, we have
used both the Exchange algorithm and our methodology to carry out
inference over a $40\times40$ grid. At this size it is not possible to
use the matrix transfer method to run an ``exact'' chain. Sequential
Monte Carlo (SMC) was used to estimate ${\mathcal Z }_i(\theta)$ at
each iteration in the Roulette implementation. The estimates of the
means and the standard deviations from both methods again agreed well
(to the third or fourth decimal place). We have also carried out
inference over a $60\times60$ grid; however, it is no longer possible
to perfectly sample at this size, particularly for parameter values
near the critical value.

\subsection{The Fisher--Bingham Distribution on a Sphere}

The Fisher--Bingham distribution \citep{kent1982} is constructed by
constraining a multivariate Gaussian vector to lie on the surface of a
$d$-dimensional unit radius sphere, $S_d$. Its form is
\begin{eqnarray*}
p(\mathbf{y}|\mathbf{A}) \propto\exp\bigl\{ \mathbf{y}'\mathbf{A}
\mathbf{y} \bigr\},
\end{eqnarray*}
where $\mathbf{A}$ is a $d\times d$ symmetric matrix and, from here on,
we take $d = 3$. After rotation to principle axes, $\mathbf{A}$ is
diagonal and so the probability density can be written as
\begin{eqnarray*}
p(\mathbf{y}|\bolds\lambda) \propto\exp\Biggl\{ \sum
_{i =
1}^{d}\lambda_i y_{i}^2
\Biggr\}.
\end{eqnarray*}

This is invariant under addition of a constant factor to each $\lambda
_i$, so for identifiability we take $0 = \lambda_1 \geq\lambda_2
\geq\lambda_3$. The normalising constant, ${\mathcal Z}(\bolds\lambda)$, is given by
\begin{eqnarray*}
{\mathcal Z}(\bolds\lambda) = \int_{\mathcal S} \exp\Biggl\{
\sum_{i = 1}^{d}\lambda_i
y_{i}^2\Biggr\} \mu(d\mathbf{y}),
\end{eqnarray*}
where $\mu(d\mathbf{y})$ represents the Hausdorff measure on the
surface of a sphere. Very few papers have presented Bayesian posterior
inference over the distribution due to the intractable nature of
${\mathcal Z}(\bolds\lambda) $. However, in a recent paper,
Walker uses an auxiliary variable method \citep{walker2011posterior}
outlined in the \hyperref[sec1]{Introduction} to sample from $p(\bolds\lambda
|\mathbf{y})$. We can apply our version of the Exact-Approximate
methodology, as we can use importance sampling to get unbiased
estimates of the normalising constant.

Twenty data points were simulated using an MCMC sampler with
$\bolds\lambda= [0,0,-2]$ and posterior inference was carried
out by drawing samples from $p(\lambda_3|\mathbf{y})$, that is, it
was assumed $\lambda_1 = \lambda_2 = 0$. Our Exact-Approximate
methodology was applied using the geometric construction with Russian Roulette truncation. A uniform distribution on the surface of a sphere
was used to draw importance samples for the estimates of ${\mathcal
Z}(\bolds\lambda)$. The proposal distribution for the parameters
was Gaussian with mean given by the current value, a uniform prior on
$[-5,0]$ was set over $\lambda_3$, and the chain was run for 20,000
iterations. Walker's auxiliary variable technique was also implemented
for comparison using the same prior but with the chain run for 200,000
samples and then the chain thinned by taking every 10th sample to
reduce strong autocorrelations between samples. In each case the final
10,000 samples were then used for Monte Carlo estimates.

%
\begin{table}[b]
\tabcolsep=0pt
\caption{Estimates of the posterior mean and standard deviation of the
posterior distribution using roulette and Walker's method for the
Fisher--Bingham distribution. An estimate of the effective sample size
(ESS) is also shown based on 10,000 MCMC samples}\label{tableBingham}
\begin{tabular*}{\tablewidth}{@{\extracolsep{\fill}}@{}lcc@{}}
\hline
 & \textbf{Roulette} & \textbf{Walker}\\
\hline
Estimate of mean & $-$2.377 & $-$2.334 \\
Estimate of standard deviation & 1.0622 & 1.024\\
ESS & 1356 & 212\\
\hline
\end{tabular*}
\end{table}

%
\begin{figure*}[t]

\includegraphics{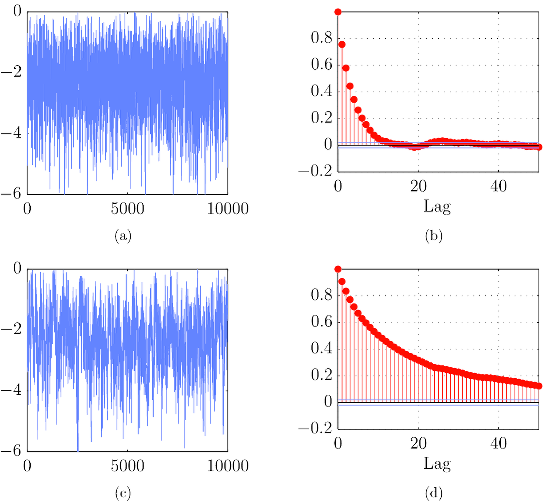}

\caption{Sample traces and autocorrelation plots for the
Fisher--Bingham distribution for the geometric tilting with Russian Roulette truncation [\textup{(a)} and \textup{(b)}] and Walker's auxiliary variable
method [\textup{(c)} and \textup{(d)}].}
\label{figBinghamACF}
\end{figure*}

In the Russian Roulette method, six negative estimates were observed in
10,000 estimates. The estimates of the mean and standard deviation of
the posterior agree well (Table~\ref{tableBingham}), however, the
effective \mbox{sample} size and autocorrelation of the Russian Roulette
method are superior as seen in Figure~\ref{figBinghamACF}. Note that
it is also possible to get an upper bound on the importance sampling
estimates for the Fisher--Bingham distribution. If we change our
identifiability constraint to be $0 = \lambda_1 \leq\lambda_2 \leq
\lambda_3$, we now have a convex sum in the exponent which can be
maximised by giving unity weight to the largest $\lambda$, that is,
$\sum_{i=1}^{d} \lambda_i y_i^2 < \lambda_{\max}$. We can compute
$\widetilde{\mathcal Z}(\bolds\theta)$ as $1/N \sum_n \exp
(\lambda_{\max})/g(y_n)$, where $g(y)$ is the importance distribution.

\section{The Limits of Exact Approximate Methods: The Ozone Data Set}
\label{secOzone}

In the previous sections of this paper we have combined various ideas
from both the Statistics and Physics literature to suggest a
pseudo-marginal MCMC scheme for doubly-intractable distributions.
Further, we have shown experimentally that this method can be
implemented in a range of Bayesian inference problems. We now turn our
attention to a case where this methodology runs into difficulty.

It is tempting to think that the method could be used to tackle very
large problems in which, for example, the likelihood requires the
computation of the determinant of a very large matrix. For many
problems the matrix in question is so large that it is not possible to
compute its Cholesky decomposition, and hence not possible to compute
the determinant. As methods are available to produce unbiased estimates
of the log determinant (\cite{bai1996some}, \cite{aune2014parameter}), the idea
would be to write the determinant, ${\mathcal{D}(\theta)}$, as
${\mathcal{D}(\theta)} = \exp(\log{\mathcal{D}(\theta)})$ and
then use the Maclaurin series expansion of the exponential function in
which each term can be estimated unbiasedly. The infinite series can
then be unbiasedly truncated using Russian Roulette methods and the
overall estimate plugged into a pseudo-marginal MCMC scheme.
Theoretically, this is an exact scheme to sample from the posterior of
such a model; however, upon closer inspection, there are several
practical difficulties associated with such an approach, namely, that
it is not possible to realise a fully unbiased estimate of the log
determinant. For exposition purposes, we now describe a specific
example of a posterior for which it is difficult if not impossible to
realise an unbiased estimate of the likelihood. In particular, we
consider the total column ozone data set that has been used many times
in the literature to test algorithms for large spatial problems
(\cite{Cressie2008Fixed}, \cite{Jun2008}, \cite{bolin2009}, \cite{aune2014parameter}, \cite{eidsvik2013estimation}).
This data set is representative of the types of problems for which
exact Markov chain Monte Carlo is considered infeasible. While large,
this data set is still of a size to run exact inference on and it
serves as an interesting example of a problem in which the methods
discussed in this paper break down. Full details and an implementation
can be found at \surl{http://www.ucl.ac.uk/roulette}.

We begin by describing the model and inference problem, and then
suggest reasons why an application of the pseudo-marginal approach may
run into difficulties. We close by describing results we were able to
obtain and giving pointers to alternative approaches for similar problems.

\subsection{The Model}

The data, which is shown in Figure~\ref{figozonedata}, consists of $N
={}$173,405 ozone measurements gathered by a satellite with a passive
sensor that measures back-scattered light \citep{Cressie2008Fixed}.
While a full analysis of this data set would require careful modelling
of both the observation process and the uncertainty of the field, for
the sake of simplicity, we will focus on fitting a stationary model.

%
\begin{figure*}[t]

\includegraphics{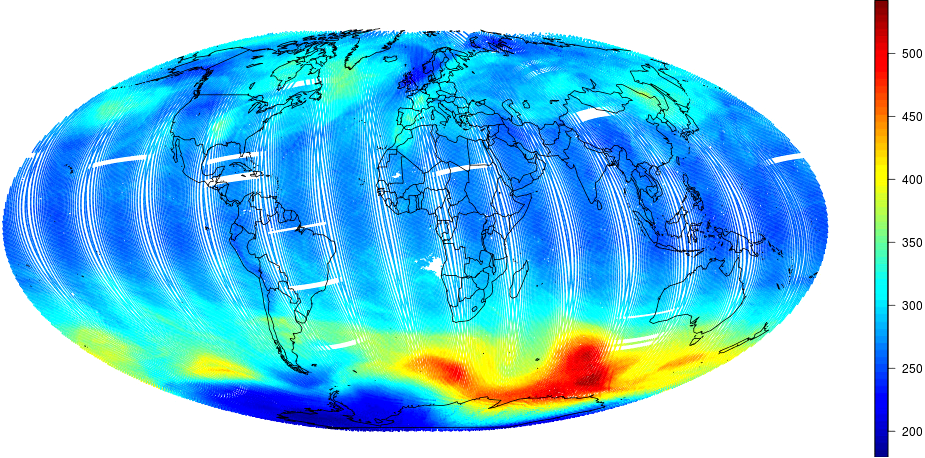}

\caption{The total ozone column data set, aligned with a map of the world.}
\label{figozonedata}
\end{figure*}

We model the data using the following three-stage hierarchical model:
%
%
\begin{eqnarray}
\label{eqnhierarchicalmodel} y_i | \mathbf{x}, \kappa, \tau&\sim&\mathcal
{N}\bigl(\mathbf{Ax}, \tau^{-1} \mathbf{I}\bigr),\nonumber
\\
\mathbf{x} | {\kappa} &\sim&\mathcal{N}\bigl(\mathbf{0}, \mathbf{Q}(
\kappa)^{-1}\bigr),
\\
\kappa&\sim&\log_2\mathcal{N}(0,100),\quad\tau\sim
\log_2\mathcal{N}(0,100),\hspace*{-15pt}\nonumber
\end{eqnarray}
where $\mathbf{Q}(\kappa)$ is the precision matrix of a Mat\'{e}rn
stochastic partial differential equation (SPDE) model defined on a
fixed triangulation of the globe and $\mathbf{A}$ is a matrix that
evaluates the piecewise linear basis functions in such a way that
$x(s_i) = [\mathbf{Ax}]_i$. The parameter $\kappa$ controls the range
over which the correlation between two values of the field is
essentially zero \citep{Lindgren2011}. The precision matrix $\mathbf
{Q}(\kappa)$ is sparse, which allows both for low-memory storage and
for fast matrix-vector products.

In this paper, the triangulation over which the SPDE model is defined
has $n ={}$196,002 vertices that are spaced regularly around the globe,
allowing piecewise linear spatial prediction. As the observation
process is Gaussian, a straightforward calculation shows that
%
%
\begin{eqnarray}
\mathbf{x} | \mathbf{y},\kappa,\tau
&\sim& N \bigl( \tau\bigl(\mathbf
{Q}(\kappa) +
\tau\mathbf{A}^T\mathbf{A}\bigr)^{-1}\mathbf{A}^T
\mathbf{y},
\nonumber\\[-8pt]\\[-8pt]\nonumber
&& \bigl(\mathbf{Q}(\kappa) + \tau\mathbf{A}^T\mathbf{A}
\bigr)^{-1} \bigr).
\end{eqnarray}

Given the hierarchical model in \eqref{eqnhierarchicalmodel}, we are
interested in the parameters $\kappa$ and $\tau$ only. To this end,
we sample their joint posterior distribution given the observations
$\mathbf{y}$, marginalised over the latent field $\mathbf{x}$, which gives
$
\pi(\kappa, \tau|\mathbf{y})\propto\pi(\mathbf{y}|\kappa, \tau
)\pi(\kappa)\pi(\tau)$.
To compute this expression, we need the marginal likelihood $\pi
(\mathbf{y}|\kappa, \tau)$, which in this case is available
analytically since $\pi(\mathbf{y}|\mathbf{x},\tau)$ and $\pi(\mathbf
{x}|\kappa)$ are both Gaussian,
%
%
\begin{eqnarray}\label{eqnozonelikelihood}
\pi(\mathbf{y}|\kappa, \tau)&=&\int\pi(\mathbf{y}|\mathbf{x}, \kappa,
\tau)\pi(
\mathbf{x}|\kappa)\,d\mathbf{x}
\nonumber\\[-8pt]\\[-8pt]\nonumber
& =&\mathcal{N}\bigl(\mathbf{0}, \tau^{-1} \mathbf{I}
+ \mathbf{A}\mathbf{Q}(\kappa)^{-1}\mathbf{A}^T\bigr).
\end{eqnarray}

Using the matrix inversion lemma to avoid storing nonsparse matrices,
the log marginal likelihood is
\begin{eqnarray}
\label{eqnozoneloglikelihood} 2\mathcal{L}(\theta)&:=&2\log\pi(\mathbf
{y}|\kappa, \tau)\nonumber
\\
&=&
C + \log\bigl(\det\bigl(\mathbf{Q}(\kappa)\bigr)\bigr)+N\log(\tau)
\nonumber\\[-8pt]\\[-8pt]\nonumber
&&{} -\log\bigl(
\det\bigl(\mathbf{Q}(\kappa)+\tau\mathbf{A}^T\mathbf{A}\bigr)\bigr)\nonumber
\\
&&{} - \tau\mathbf{y}^T\mathbf{y}+\tau^2
\mathbf{y}^T\mathbf{A}\bigl(\mathbf{Q}(\kappa)+\tau
\mathbf{A}^T\mathbf{A}\bigr)^{-1}\mathbf{A}^T
\mathbf{y}.\nonumber
\end{eqnarray}

\subsection{Likelihood Estimation and Russian Roulette}

In order to apply a pseudo-marginal MCMC scheme, we require an unbiased
estimate of \eqref{eqnozonelikelihood}, for which we first need to
compute unbiased estimates of the log-likelihood \eqref
{eqnozoneloglikelihood}. Those are then plugged into a Russian Roulette
truncated Maclaurin expansion of the exponential function, $\exp
(\mathcal{L}(\theta))=\sum_{n=0}^{\infty}\frac{\mathcal{L}(\theta
)^n}{n!}$ [after replacing each $\mathcal{L}(\theta)$ with an
unbiased estimate], to obtain the required unbiased estimate of the
overall Gaussian likelihood \eqref{eqnozonelikelihood}.

To construct an unbiased estimator of \eqref{eqnozoneloglikelihood},
the main challenge is to estimate $\log(\det(\mathbf{Q}))$. We note that
%
%
\begin{equation}
\qquad \log\bigl(\det(\mathbf{Q})\bigr) =\tr\bigl(\log(\mathbf{Q})\bigr)=
\mathbb{E}_{\mathbf{z}}\bigl(\mathbf{z}^T\log(\mathbf{Q})
\mathbf{z}\bigr), \label{eqnlogdetmontecarlo}
\end{equation}
where $\mathbf{z}$ is a vector of i.i.d. centred, unit variance random
variables \citep{bai1996some}.
Therefore, an unbiased estimator of the log-determi\-nant can be
constructed through Monte Carlo estimates of the expectation with
respect to the distribution of $\mathbf z$. \citet{aune2014parameter} used
rational approximations and Krylov subspace methods to compute each
$\log(\mathbf{Q}) \mathbf{z}$ in \eqref{eqnlogdetmontecarlo} to machine
precision, and they introduced a graph colouring method that massively
reduces the variance in the Monte Carlo estimator. This approach is
both massively parallel and requires a low-memory overhead, as only
$\mathcal{O}(1)$ large vectors need to be stored on each processor.

However, as already mentioned, several issues are foreseeable when
applying the pseudo-marginal MCMC scheme to the posterior.
We emphasize two main points here:


%
\begin{longlist}[2.]
\item[1.]\emph{Numerical linear algebra}: In order to compute estimates
of the log-likelihood \eqref{eqnozoneloglikelihood}, we need to solve
a number of sparse linear systems. More precisely, we apply the
methodology of \citet{aune2014parameter}, which reduces computing each
log-determinant to solving a family of shifted linear equations for
each of the $\log(\mathbf{Q}) \mathbf{z}$ in \eqref{eqnlogdetmontecarlo}.
In addition, we need to solve the matrix inversions in \eqref
{eqnozoneloglikelihood}. Note that each sparse linear system is
independent and may be solved on its own separate computing node.
Speed of convergence for solving these sparse linear systems largely
depends on the condition number of the underlying matrix---the ratio
of the largest and the smallest eigenvalues. In this example, the
smallest eigenvalue of $\mathbf{Q}(\kappa)$ is arbitrarily close to zero,
which catastrophically affects convergence of the methods described in
\citet{aune2014parameter}. We can partially overcome these practical
issues by regularising the matrix's smallest eigenvalue via adding a
small number to the diagonal, shrinking the condition number using
preconditioning matrices for the conjugate gradient, and setting a
large iteration limit for the linear solvers. These convergence
problems are typical when considering spatial models, as the
eigenvalues of the continuous precision operator are unbounded. This
suggests a fundamental limitation to exact-approximate methods for
these models: it is impossible to attain full floating point precision
when solving these linear systems, and hence the resulting Markov chain
cannot exactly target the marginal posterior density $\pi(\kappa,
\tau|\mathbf{y})$. 

\item[2.] {\emph{Scaling}: A big challenge for practically implementing
the Russian Roulette step is the large amount of variability in the
estimator for \eqref{eqnozoneloglikelihood}, which is amplified by
Russian Roulette. Denote by $\widehat{{\mathcal{L}(\theta)}} - U$
the unbiased estimator of the log-likelihood in equation
(\ref{eqnozoneloglikelihood}), shifted towards a lower bound (see below)
to reduce its absolute value. When the variance of the log-determinant
estimator is large, the exponential series expansion will converge
slowly and we will need to keep a large number of terms in order to
keep the variance of the overall estimate low. We can get around this
by borrowing the idea of ``scaling-and-squaring'' from numerical
analysis \citep{Golub1996}.

%
\begin{figure*}[t]

\includegraphics{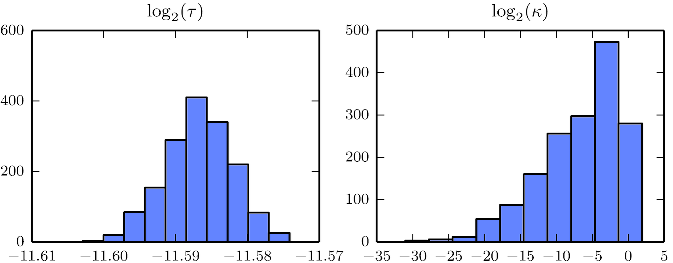}

\caption{Histograms of the marginals $p(\log_2(\tau)|\mathbf{y})$
(left) and $p(\log_2(\kappa)|\mathbf{y})$ (right).}
\label{figposteriors}
\end{figure*}

We find an integer $E\in\mathbb{N}$ with $E\approx\mathcal
{L}(\theta)-U$, for example, by averaging over a number of estimates.
We then write
%
%
\begin{eqnarray}
\qquad\exp\bigl(\mathcal{L}(\theta)-U\bigr)= \biggl(\exp\biggl(\frac{\mathcal
{L}(\theta)-U}{E}
\biggr) \biggr)^E. \label{eqnRRrescalling}
\end{eqnarray}
In order to compute an unbiased estimate for this expression, we need
to multiply $E$ unbiased estimates of $\exp\{(\mathcal{L}(\theta
)-U)/E\}$, each of which we can obtain using Russian Roulette. This is
now an easier problem since $(\mathcal{L}(\theta)-U)/E\approx1$ is
close to one. Therefore, the exponential series converges rapidly so
that we only need a few estimates for $(\mathcal{L}(\theta)-U)/E$ in
order to obtain one estimate of $\exp\{(\mathcal{L}(\theta)-U)/E\}$.
The fact that a lower bound for $\mathcal{L}(\theta)$ is unavailable
compromises unbiasedness of the estimator. In practice, this, however,
was not measurable and drastically improved run-time. }

\end{longlist}

\subsection{Results and Remarks on Approximate Schemes}

As this model is sufficiently small to (with some effort) perform exact
inference, we began by finding the exact marginal posterior $\pi
(\kappa,\tau\mid\mathbf{y})$, which is shown in Figure~\ref
{figposteriors}. The resulting density is relatively simple, which
suggests that an appropriately scaled random walk Metropolis algorithm
is sufficient for exploring it. As expected, in contrast to the other
cases examined in this paper, we found that the Russian Roulette random
walk Metropolis chain failed to converge for this problem: the chain
exhibited catastrophic sticking and therefore extremely high
autocorrelation. This is likely due to a combination of (a) our
approximations in the log-likelihood estimator (due to the ill-posed
linear systems), (b) the variation of the log-determinant estimator due
to slow convergence of the linear solvers, and (c) the bias due to
introducing the above scaling trick to the Russian Roulette scheme.


The fact that it is infeasible to realise a genuinely unbiased estimate
of the normalising term for a model of this nature and size may mean
that the Russian Roulette framework (and perhaps the entire concept of
exact-approximate methods) is not the right approach for this type of
model. We note that it has been shown previously in the literature that
there are limitations to the efficiency of the pseudo-marginal scheme.
For example, \citet{sherlock2013efficiency} established results on
optimal scaling and acceptance rate which indicate compromised
efficiency when using the scheme.
We close with the remark that compromising the ambitious goal of
performing full and exact Bayesian inference on this problem might be a
reasonable approach for practitioners who are interested in using
models of the above type for solving large-scale problems. Recently,
there has been an increased interest in approximate Markov transition
kernels that allow such trade-off between computing time and introduced
bias. Most of those methods are based on subsampling available
observations in the Big Data case (\cite{bardenet2014towards}, \cite{korattikara2014austerity}, \cite{welling2011bayesian}), and
are therefore not available for the described ozone model, where we aim
to do inference for a \emph{single} observation. Similarly, the
Exchange algorithm \citep{murray2006mcmc} is unavailable due to
sampling from the likelihood being infeasible.

Using an approximate Markov transition kernel, induced from any
approximation to the likelihood, leads to a chain whose invariant
distribution is not equal to the true marginal posterior. Recently,
\citet{alquier2014noisy} reviewed and analysed many cases of such
approximate MCMC algorithms.
A weak form of convergence is given by \citet{alquier2014noisy} (Theorem~2.1), which states that a Markov chain induced by an approximate
transition kernel which approaches its exact version in the limit has
an invariant distribution and this converges to the desired
distribution as the kernel converges to the exact kernel under certain
conditions. Theoretically, it is possible to apply this \mbox{approach} to the
ozone example, as we can get a biased estimator for the log-likelihood
via avoiding the Russian Roulette, and this kernel becomes exact when
we use a very large number of iterations in the linear solvers.
However, the slow convergence of the solvers remains a problem and in
fact leads to such large variation in the log-likelihood estimate that
again the chain catastrophically sticks. It would seem, for the moment,
that further approximation is required in order to carry out Bayesian
inference. For example, in \citet{shaby2014open}, it is suggested that
a function other than the likelihood, for example, a composite
likelihood, can be used in an MCMC scheme to obtain samples from a
``quasi-posterior'' which can then be rotated and scaled to give
asymptotically valid estimates.


\section{Discussion and Conclusion}\label{secdisc}

The capability to perform pseudo-marginal MCMC on a wide class of
doubly-intractable distributions has been reviewed and established in
this paper. The methods described are not reliant on the ability to
simulate exactly from the underlying model, only on the availability of
unbiased estimates of the inverse of a normalising term,
which makes them applicable to a wider range of problems than has been
the case to date.

The development of this method, which returns an unbiased estimate of
the target distribution, is based on the stochastic truncation of a
series expansion of the desired density. If the intractable likelihood
is composed of a bounded function and nonanalytic normalising term,
then the proposed methodology can proceed to full MCMC with no further
restriction. However, in the more general case, where an unbounded
function forms the likelihood, then the almost sure guarantee of \emph
{positive} unbiased estimates is lost. The potential bias induced due
to this lack of strict positivity is dealt with by adopting a scheme
employed in the QCD literature where an absolute measure target
distribution is used in the MCMC and the final Monte Carlo estimate is
``sign corrected'' to ensure that expectations with respect to the
posterior are preserved. The inflation of the Monte Carlo error in such
estimates is a function of the severity of the sign problem and this
has been characterised in our work. What has been observed in the
experimental evaluation is that, for the examples considered, the sign
problem is not such a practical issue when the variance of the
estimates of the normalising terms is well controlled and this has been
achieved by employing \mbox{Sequential} Monte Carlo Sampling in some of the
examples. Hence, one of the areas for future work is efficient
estimators of the normalising term, which can be either unbiased or
merely consistent. Indeed, for the total column ozone data set, it is
not possible at present to realise a completely unbiased estimate of
$\log(\det(\mathbf{Q}))$, as is required for the pseudo-marginal
methodology. The inherent computational parallelism of the methodology,
due to it only requiring a number of independent estimates of
normalising constants, indicates that it should be possible to
implement this form of inference on larger models than currently
possible, however it is also clear that there is some limit to how much
the method can be scaled up. For the time being, approximate methods
described in Section~\ref{secreview} can be used for very large-scale
models, for example, analytic approximations to the posterior \citep
{rue2009approximate} or ABC \citep{moores2014pre} could be used.

It has been shown \citep{jacob2013non} that it is not possible to
realise strictly positive estimates of the target distribution using
the series expansions described in this paper, unless the estimates of
the normalising term lie in a bounded interval. In its most \mbox{general}
representation it is recognised that the sign problem is NP-hard,
implying that a practical and elegant solution may remain elusive for
some time to come. However, other ideas from the literature, such as
the absolute measure approach \citep{lin2000noisy}, can be used to
tackle the sign problem. The methodology described in this paper
provides a general scheme with which Exact-Approximate MCMC for
Bayesian inference can be deployed on a large class of statistical
models. This opens up further opportunities in statistical science and
the related areas of science and engineering that are dependent on
simulation-based inference schemes.

\begin{appendix}\label{app}
\section{Russian Roulette}

Consider approximating the sum $S=\sum_{k\geq0} \alpha_k$ assumed
finite. Let $\tau$ denote a finite random time taking positive
integer values such that $p_n\eqdef\PP(\tau\geq n)>0$ for all $n\geq
0$. The fact that $\tau$ is finite almost surely means that
%
%
\begin{equation}
\label{cond1} \PP(\tau=\infty)=\lim_{n\to\infty} p_n=0.
\end{equation}

We consider the weighted partial sums $S_0=\alpha_0$, and for $k\geq1$,
\[
S_k=\alpha_0 + \sum_{j=1}^k
\frac{\alpha_j}{p_j}.
\]

For completeness, we set $S_\infty=\infty$. The Russian Roulette
random truncation approximation of $S$ is
\[
\hat S=S_\tau.
\]

If $\tau$ can be easily simulated and the probabilities $p_n$ are
available then $\hat S$ can be computed. The next result states that
$\hat S$ is an unbiased estimator of $S$.
%
\begin{proposition}\label{rouletteprop1}
The random variable $\hat S$ has finite
expectation, and $\PE(\hat S)=S$.
\end{proposition}

\begin{pf}
Set $\bar S_0=|\alpha_0|$, and $\bar S_k=|\alpha_0| + \sum_{j=1}^k
|\alpha_j|/ p_j$. Then for all $n\geq1$
\begin{eqnarray*}
&& \sum_{k=0}^n |S_k|\PP(
\tau=k)
\\
&&\quad \leq \sum_{k=0}^n \bar
S_k\PP(\tau=k) = \sum_{k=0}^n
\bar S_k (p_k-p_{k+1} )
\\
&&\quad = \bar S_0 p_0 + \sum_{k=1}^n
(\bar S_k-\bar S_{k-1} ) p_k
\\
&&\qquad{}+ \sum
_{k=1}^n \bar S_{k-1} p_k -\sum
_{k=0}^n \bar S_k
p_{k+1}
\\
&&\quad =\sum_{k=0}^n|\alpha_k|-
\bar S_np_{n+1}\leq\sum_{k=0}^n|
\alpha_k|.
\end{eqnarray*}

Since $\sum_n|\alpha_n|<\infty$, we conclude that $\sum_n|S_n|\PP
(\tau=n)<\infty$, hence $\PE(|\hat S|)<\infty$. A similar
calculation as above gives for all $n\geq1$,
\[
\sum_{k=0}^n S_k\PP(\tau=k) =
\sum_{k=0}^n\alpha_k
-S_np_{n+1}.
\]
By Kronecker's lemma $\lim_{n\to\infty} p_n S_n=0$, and $|p_{n+1}
S_n|=(p_{n+1}/p_n)p_n|S_n|\leq p_n|S_n|\to0$, as\vspace*{1pt} $n\to\infty$. We
conclude that $\PE(\hat S)=\sum_{k=0}^\infty S_k\PP(\tau=k) = \sum
_{k=0}^\infty\alpha_k$.
\end{pf}

This random truncation approximation of the series $\sum_n\alpha_n$
is known in the Physics literature as Russian Roulette. It has been
re-derived apparently independently by \citet{mcleish2011general}. In
the Physics literature it is common to choose $\tau$ as a stopping
time of the form
\[
\tau=\inf\{k\geq1: U_k\geq q_k \},
\]
where $\{U_j, j\geq1\}$ are i.i.d. $\mathcal{U}(0,1)$, $q_j\in
(0,1]$ and $\hat{S} = S_{\tau-1}$. In this case $p_n=\prod_{j=1}^{n-1}
q_j$. The random time $\tau$ can be thought as the running
time of the algorithm. It is tempting to choose $\tau$ such that the
Russian Roulette terminates very quickly. The next result shows that
the resulting variance will be high, possibly infinite.

%
\begin{proposition}\label{propvar}
If
\[
\sum_{n\geq1}\frac{|\alpha_n|}{p_n}\sup
_{j\geq n}\Biggl\llvert\sum_{\ell=n}^j
\alpha_\ell\Biggr\rrvert<\infty,
\]
then $\textsf{Var}(\hat S)<\infty$ and
\[
\textsf{Var}(\hat S)=\alpha_0^2 +\sum
_{n\geq1} \frac{\alpha
_n^2}{p_n} + 2\sum_{n\geq1}
\alpha_n S_{n-1} -S^2.
\]

If $\{\alpha_n\}$ is a sequence of nonnegative numbers and $\sum_{n\geq
1} \alpha_n S_{n-1}=\infty$, then $\textsf{Var}(\hat
S)=\infty$.
\end{proposition}

\begin{pf}
$\textsf{Var}(\hat S)=\PE(\hat S^2)-S^2$. So it\vspace*{1pt} suffices to work with
$\PE(\hat S^2)$. $\PE(\hat S^2)=\sum_{k=0}^\infty S_k^2\PP(\tau
=k)=\lim_{n\to\infty}\sum_{k=0}^n S_k^2\PP(\tau=k)$. For any
$n\geq1$, we use the same telescoping trick used in Proposition \ref
{rouletteprop1} to get
%
%
\begin{eqnarray}
\label{propvareq1}
\qquad&& \sum_{k=0}^n
S_k^2\PP(\tau=k)\nonumber
\\
&&\quad =\sum_{k=0}^nS_{k-1}^2
(p_k-p_{k+1} )
\\
&&\quad = \alpha_0^2 + \sum_{k=1}^n
\frac{\alpha_k^2}{p_k} +2\sum_{k=1}^n
\alpha_k S_{k-1} -S_n^2p_{n+1}.\nonumber
\end{eqnarray}

By Jensen's inequality $S_n^2\leq(\sum_{k=1}^n p_k^{-1}
)\times  (\sum_{k=1}^n p_k^{-1}\alpha_k^2 )$. Hence, using
Kronecker's lemma, we see that
%
%
\begin{eqnarray}\label{propvareq2}
p_{n+1}S_n^2&\leq&
p_nS_n^2\nonumber
\\
&\leq&\Biggl(p_n\sum
_{k=1}^n \frac
{1}{p_k} \Biggr) \Biggl(
\sum_{k=1}^n \frac{\alpha_k^2}{p_k} \Biggr)
\\
&=& o
\Biggl(\sum_{k=1}^n \frac{\alpha_k^2}{p_k}
\Biggr),\quad\mbox{as } n\to\infty,\nonumber
\end{eqnarray}
so it suffices to show that the sequence $\sum_{k=1}^n\frac{\alpha
_k^2}{p_k} +\sum_{k=1}^n \alpha_k S_{k-1}$ is bounded. But
\begin{eqnarray*}
&& \Biggl\llvert\sum_{j=1}^n
\frac{\alpha_k^2}{p_k} +\sum_{k=1}^n
\alpha_k S_{k-1}\Biggr\rrvert
\\
&&\quad =\Biggl\llvert
\alpha_0\sum_{j=0}^n
\alpha_j + \sum_{j=1}^n
\frac{\alpha_j}{p_j} \Biggl(\sum_{k=j}^n
\alpha_k \Biggr)\Biggr\rrvert
\\
&&\quad \leq |\alpha_0|\sum_{j\geq0}|
\alpha_j| + \sup_n\sum
_{j=1}^n \frac{|\alpha_j|}{p_j}\Biggl\llvert\sum
_{k=j}^n\alpha_k\Biggr\rrvert,
\end{eqnarray*}
and the two terms on the right-hand side are bounded under the stated
assumptions. Therefore the series $\sum_n S_n^2\PP(\tau=n)$ is
summable and the variance formula follows by taking the limit as $n\to
\infty$ in (\ref{propvareq1}).

To establish the rest of the proposition, we deduce from (\ref
{propvareq2}) that for $n$ large enough
\[
\sum_{k=1}^n S_k^2
\PP(\tau=k) \geq\alpha_0^2 + 2\sum
_{k=1}^n \alpha_k S_{k-1},
\]
which easily implies the statement.
\end{pf}

%
\begin{remark}
As an example, for a geometric sequence $\alpha_i=\alpha^i$ for
$\alpha\in(0,1)$, and we choose $q_i=q$ for some $q\in(0,1)$, then
for $\alpha^2/q<1$, the condition of Proposition \ref{propvar} are
satisfied\vspace*{1pt} and $\textsf{Var}(\hat S)<\infty$. If $q>\alpha^2$ the variance is
infinite. The average computing time of the algorithm is $\PE(\hat
\tau)=\frac{1}{1-q}$. Although this variance/computing speed
trade-off can be investigate analytically, a rule of thumb that works
well in simulations is to choose $q=\alpha$.
\end{remark}

\section{Computing Absolute Measure Expectations}\label{secabsmeas}
Let $(\Xset,\B)$ denotes a general measure space with a reference
sigma-finite measure $dx$. Let $\pi: \Xset\to\rset$ a function
taking possibly negative values such that $\int|\pi(x)|\,dx<\infty$.
We assume that $\int\pi(x)\,dx>0$ and we wish to compute the quantity
\[
I=\frac{\int h(x)\pi(x)\,dx}{\int\pi(x)\,dx},
\]
for some measurable function $h: \Xset\to\rset$ such that $\int
|h(x)\pi(x)|\,dx<\infty$. We introduce $\sigma(x)= \textsf{sign}(\pi
(x))$, and $p(x)=\frac{|\pi(x)|}{\int|\pi(x)|\,dx}$. Thus $p$ is a
probability density on $\Xset$. Suppose that we can construct an
ergodic Markov chain $\{X_n, n\geq0\}$ with invariant distribution
$p$, for instance using the Metropolis--Hastings algorithm. An
importance sampling-type estimate for $I$ is given by
\[
\hat I_n=\frac{\sum_{k=1}^n \sigma(X_k)h(X_k)}{\sum_{k=1}^n \sigma(X_k)}.
\]
$\hat I_n$ has the following properties.
%
\begin{proposition}

\begin{longlist}[1.]
\item[1.] If the Markov chain $\{X_n, n\geq0\}$ is phi-irreducible and
aperiodic, then $\hat I_n$ converges almost surely to $I$ as $n\to
\infty$.

\item[2.] Suppose that $\{X_n, n\geq0\}$ is geometrically ergodic and
$\int|h(x)|^{2+\varepsilon}p(x)\,dx<\infty$ for some $\varepsilon>0$. Then
\[
\sqrt{n} (\hat I_n-I )\stackrel{\textsf{w}} {\to} \mathbf{N}
\bigl(0,\sigma^2(h)\bigr),
\]
where
\[
\sigma^2(h)=
\frac{C_{11}+
I^2C_{22} -2IC_{12}}{r^2},
\]
and
\begin{eqnarray*}
C_{11}&=&\textsf{Var}_p \bigl(\{h\sigma\}(X) \bigr)
\\
&&{}\cdot\sum
_{j=-\infty
}^\infty\textsf{Corr}_p
\bigl(\{h\sigma\}(X),P^{|j|}\{h\sigma\} (X) \bigr),
\\
C_{22}&=&\textsf{Var}_p\bigl(\sigma(X)\bigr)\sum
_{j=-\infty}^\infty\textsf{Corr}_p \bigl(
\sigma(X),P^{|j|}\sigma(X) \bigr),
\\
C_{12}&=&\frac{1}{2}\sqrt{\textsf{Var}_p\bigl(\{h
\sigma\}(X)\bigr)\textsf{Var}_p\bigl(\sigma(X)\bigr)}
\\
&&\cdot\Biggl[\sum_{j=-\infty}^\infty
\textsf{Corr}_p \bigl(\{ h\sigma\}(X),P^{|j|}\sigma(X)
\bigr)
\\
&&{}+\sum_{j=-\infty}^\infty\textsf{Corr}_p
\bigl(\sigma(X),P^{|j|}\{h\sigma\}(X) \bigr) \Biggr].
\end{eqnarray*}
\end{longlist}
\end{proposition}

\begin{pf}
Part (1) is a straightforward application of the law of large numbers
for the Markov chain $\{X_n, n\geq0\}$: as $n\to\infty$, $\hat I_n$
converges almost surely to
\[
\frac{\int\sigma(x)h(x)p(x)\,dx}{\int\sigma(x)p(x)\,dx}= \frac{\int
h(x)\pi(x)\,dx}{\int\pi(x)\,dx}=I.
\]

A bivariate central limit theorem using the Cramer--Wold device gives that
\begin{eqnarray*}
&& \sqrt{n}\pmatrix{\displaystyle \frac{1}{n}\sum_{k=1}^n
\sigma(X_k)h(X_k) - rI
\cr
\displaystyle\frac{1}{n}\sum
_{k=1}^n \sigma(X_k) - r}
\\
&&\quad \stackrel{
\textsf{w}} {\to} \pmatrix{Z_1
\cr
Z_2}\sim\mathbf{N} \biggl[
\pmatrix{0
\cr
0}, \pmatrix{C_{11} & C_{12}
\cr
C_{12}
& C_{22}} \biggr],
\end{eqnarray*}
where $C_{11}, C_{12}$ and $C_{22}$ are as given above.


By the delta method, it follows that $\sqrt{n}(\hat I_n-I)\stackrel
{\textsf{w}}{\to} \frac{Z_1-I Z_2}{r}\sim\textbf{N} (0,\frac
{C_{11}+I^2C_{22}-2IC_{12}}{r^2} )$.
\end{pf}

We can roughly approximate the asymptotic variance $\sigma^2(h)$ as
follows. Suppose for simplicity that the Markov chain is reversible, so that
\begin{eqnarray*}
C_{12} &=& \sqrt{\textsf{Var}_p\bigl(\{h\sigma\}(X)\bigr)
\textsf{Var}_p\bigl(\sigma(X)\bigr)}
\\
&&{}\cdot \sum
_{j=-\infty}^\infty\textsf{Corr}_p \bigl(\{h
\sigma\} (X),P^{|j|}\sigma(X) \bigr).
\end{eqnarray*}
Assume also that the mixing of the Markov chain is roughly the
same across all the functions:
\begin{eqnarray*}
&& \sum_{j=-\infty}^\infty\textsf{Corr}_p
\bigl(\{h\sigma\} (X),P^{|j|}\{h\sigma\}(X) \bigr)
\\
&&\quad =\sum
_{j=-\infty}^\infty\textsf{Corr}_p \bigl(\sigma
(X),P^{|j|}\sigma(X) \bigr)
\\
&&\quad =\sum_{j=-\infty}^\infty\frac{\textsf{Corr}_p (\sigma
(X),P^{|j|}\{h\sigma\}(X) )}{\textsf{Corr}_p (\sigma(X),\{
h\sigma\}(X) )}\equiv
V,
\end{eqnarray*}
where we also assume that $\textsf{Corr}_p (\{h\sigma\}
(X),\sigma(X) )\neq0$. Therefore
\begin{eqnarray*}
\frac{\sigma^2(h)}{V}&\approx&
\bigl(\textsf{Var}_p\bigl(\{h\sigma\}
(X)\bigr) + I^2\textsf{Var}_p\bigl(\sigma(X)\bigr)
\\
&&{} -2 I \textsf{Cov}_p \bigl(\{
h\sigma\}(X),\sigma(X) \bigr)\bigr)
/r^2
\\
&=&\frac{r\check\pi
(h^2\sigma) + I^2-2I r \check\pi(h\sigma)}{r^2},
\end{eqnarray*}
where $\check\pi=\pi/\int\pi$, and $\check\pi(f)=\int f(x)\check
\pi(x)\,dx$. By a Taylor approximation of $(h,\sigma)\mapsto h^2\sigma
$ around $(\check\pi(h), \check\pi(\sigma))$, it comes easily that
$r\check\pi(h^2\sigma)=\check\pi(h^2)+ 2Ir\check\pi(h\sigma
)-2I^2$, so that
\[
\sigma^2(h)\approx\bigl(\check\pi\bigl(h^2
\bigr)-I^2 \bigr) \times\frac{V}{r^2}.
\]
Thus\vspace*{1pt} a quick approximation of the Monte Carlo variance of $\hat I_n$ is
given by
\begin{eqnarray*}
&& \frac{1}{n}\times\biggl\{\frac{\sum_{k=1}^nh^2(X_k)\sigma
(X_k)}{\sum_{k=1}^n\sigma(X_k)}
\\
&&\quad{}- \biggl(\frac{\sum_{k=1}^nh(X_k)\sigma
(X_k)}{\sum_{k=1}^n\sigma(X_k)}
\biggr)^2 \biggr\}
\\
&&\quad{}\cdot \frac{\hat V}{ \{\sfrac{1}{n}\sum_{k=1}^n\sigma
(X_k) \}^2},
\end{eqnarray*}
where $\hat V$ is an estimate of the common autocorrelation sum. For
example $\hat V$ can be taken as the lag-window estimate of $\sum
_{j=-\infty}^\infty\textsf{Corr}_p (\{h\sigma\}(X),\break P^{|j|}\{
h\sigma\}(X) )$.

The quantity $\frac{1}{n}\sum_{k=1}^n\sigma(X_k)$ which estimates
$r$ is indicative of the severity of the issue of returning negative
estimates. The smaller $r$, the harder it is to estimate $I$ accurately.
\end{appendix}

\section*{Acknowledgements}
Anne-Marie Lyne is supported by UCL Systems Biology.
Mark Girolami is most grateful to Arnaud Doucet, David Barber,
Christian Robert, Nicolas Chopin and Gareth Roberts for numerous
motivating discussions regarding this work.
Mark Girolami is supported by the
UK Engineering and Physical Sciences Research Council \mbox{(EPSRC)} via the
Established Career Research Fellowship EP/J016934/1 and the Programme
Grant \textit{Enabling Quantification of Uncertainty for Large-Scale
Inverse Problems}, EP/K034154/1, \surl{http://www.warwick.ac.uk/equip}.
He also gratefully acknowledges support from a Royal Society Wolfson
Research Merit Award. Yves Atchad\'e is supported by the NSF on grant NSF-SES
1229261. Heiko Strathmann is supported by the Gatsby Charitable Foundation.
Daniel Simpson is
supported by CRiSM (Warwick).


%
%

%

\begin{thebibliography}{88}

\bibitem[\protect\citeauthoryear{Adams, Murray and MacKay}{2009}]{adams2009nonparametric}
%
\begin{bmisc}[auto]
\bauthor{\bsnm{Adams},~\bfnm{R.~P.}\binits{R.~P.}},
\bauthor{\bsnm{Murray},~\bfnm{I.}\binits{I.}}
\AND
\bauthor{\bsnm{MacKay},~\bfnm{D.~J.}\binits{D.~J.}}
(\byear{2009}).
\bhowpublished{Nonparametric bayesian density modeling with gaussian processes.
Preprint. Available at \arxivurl{arXiv:0912.4896}.}
\end{bmisc}
%
\bptok{imsref}%
\endbibitem

\bibitem[\protect\citeauthoryear{Alquier et~al.}{2014}]{alquier2014noisy}
%
\begin{bmisc}[auto:parserefs-M02]
\bauthor{\bsnm{Alquier},~\bfnm{P.}\binits{P.}},
\bauthor{\bsnm{Friel},~\bfnm{N.}\binits{N.}},
\bauthor{\bsnm{Everitt},~\bfnm{R.}\binits{R.}} \AND
\bauthor{\bsnm{Boland},~\bfnm{A.}\binits{A.}}
(\byear{2014}).
\bhowpublished{Noisy Monte Carlo: Convergence of Markov chains with
approximate transition kernels.
Preprint. Available at \arxivurl{arXiv:1403.5496}.}
\end{bmisc}
%
\bptok{imsref}%
\endbibitem

\bibitem[\protect\citeauthoryear{Andrieu and
Roberts}{2009}]{andrieu2009pseudo}
%
\begin{barticle}[mr]
\bauthor{\bsnm{Andrieu},~\bfnm{Christophe}\binits{C.}} \AND
\bauthor{\bsnm{Roberts},~\bfnm{Gareth~O.}\binits{G.~O.}}
(\byear{2009}).
\btitle{The pseudo-marginal approach for efficient {M}onte {C}arlo
computations}.
\bjournal{Ann. Statist.}
\bvolume{37}
\bpages{697--725}.
\bid{doi={10.1214/07-AOS574}, issn={0090-5364}, mr={2502648}}
\end{barticle}
%
\bptok{imsref}%
\endbibitem

\bibitem[\protect\citeauthoryear{Andrieu and
Vihola}{2014}]{andrieu2014establishing}
%
\begin{bmisc}[auto]
\bauthor{\bsnm{Andrieu},~\bfnm{Christophe}\binits{C.}} \AND
\bauthor{\bsnm{Vihola},~\bfnm{Matti}\binits{M.}}
(\byear{2014}).
\bhowpublished{Establishing some order amongst exact approximations of mcmcs.
Preprint. Available at \arxivurl{arXiv:1404.6909}.}
\end{bmisc}
%
\bptok{imsref}%
\endbibitem

\bibitem[\protect\citeauthoryear{Atchad{\'e}, Lartillot and
Robert}{2013}]{atchade2008bayesian}
%
\begin{barticle}[mr]
\bauthor{\bsnm{Atchad{\'e}},~\bfnm{Yves~F.}\binits{Y.~F.}},
\bauthor{\bsnm{Lartillot},~\bfnm{Nicolas}\binits{N.}} \AND
\bauthor{\bsnm{Robert},~\bfnm{Christian}\binits{C.}}
(\byear{2013}).
\btitle{Bayesian computation for statistical models with intractable
normalizing constants}.
\bjournal{Braz. J. Probab. Stat.}
\bvolume{27}
\bpages{416--436}.
\bid{doi={10.1214/11-BJPS174}, issn={0103-0752}, mr={3105037}}
\bptnote{check volume, check pages, check year}%
\end{barticle}
%
\bptok{imsref}%
\endbibitem

\bibitem[\protect\citeauthoryear{Aune, Simpson and
Eidsvik}{2014}]{aune2014parameter}
%
\begin{barticle}[mr]
\bauthor{\bsnm{Aune},~\bfnm{Erlend}\binits{E.}},
\bauthor{\bsnm{Simpson},~\bfnm{Daniel~P.}\binits{D.~P.}} \AND
\bauthor{\bsnm{Eidsvik},~\bfnm{Jo}\binits{J.}}
(\byear{2014}).
\btitle{Parameter estimation in high dimensional {G}aussian distributions}.
\bjournal{Stat. Comput.}
\bvolume{24}
\bpages{247--263}.
\bid{doi={10.1007/s11222-012-9368-y}, issn={0960-3174}, mr={3165552}}
\end{barticle}
%
\bptok{imsref}%
\endbibitem

\bibitem[\protect\citeauthoryear{Bai, Fahey and Golub}{1996}]{bai1996some}
%
\begin{barticle}[mr]
\bauthor{\bsnm{Bai},~\bfnm{Zhaojun}\binits{Z.}},
\bauthor{\bsnm{Fahey},~\bfnm{Mark}\binits{M.}} \AND
\bauthor{\bsnm{Golub},~\bfnm{Gene}\binits{G.}}
(\byear{1996}).
\btitle{Some large-scale matrix computation problems}.
\bjournal{J. Comput. Appl. Math.}
\bvolume{74}
\bpages{71--89}.
\bid{doi={10.1016/0377-0427(96)00018-0}, issn={0377-0427}, mr={1430368}}
\end{barticle}
%
\bptok{imsref}%
\endbibitem

\bibitem[\protect\citeauthoryear{Bakeyev and
De~Forcrand}{2001}]{bakeyev2001noisy}
%
\begin{barticle}[auto:parserefs-M02]
\bauthor{\bsnm{Bakeyev},~\bfnm{T.}\binits{T.}} \AND
\bauthor{\bsnm{De Forcrand},~\bfnm{P.}\binits{P.}}
(\byear{2001}).
\btitle{Noisy Monte Carlo algorithm reexamined}.
\bjournal{Phys. Rev. D}
\bvolume{63}
\bpages{54505}.
\end{barticle}
%
\bptok{imsref}%
\endbibitem

\bibitem[\protect\citeauthoryear{Bardenet, Doucet and
Holmes}{2014}]{bardenet2014towards}
%
\begin{binproceedings}[auto:parserefs-M02]
\bauthor{\bsnm{Bardenet},~\bfnm{R.}\binits{R.}},
\bauthor{\bsnm{Doucet},~\bfnm{A.}\binits{A.}} \AND
\bauthor{\bsnm{Holmes},~\bfnm{C.}\binits{C.}}
(\byear{2014}).
\btitle{Towards scaling up Markov chain Monte Carlo: An adaptive
subsampling approach}.
In \bbooktitle{Proceedings of the 31st International Conference on
Machine Learning}
\bpages{405--413}.
\bpublisher{JMLR Workshop and Conference Proceedings}.
\end{binproceedings}
%
\bptok{imsref}%
\endbibitem

\bibitem[\protect\citeauthoryear{Beaumont}{2003}]{beaumont2003estimation}
%
\begin{barticle}[pbm]
\bauthor{\bsnm{Beaumont},~\bfnm{Mark~A.}\binits{M.~A.}}
(\byear{2003}).
\btitle{Estimation of population growth or decline in genetically
monitored populations}.
\bjournal{Genetics}
\bvolume{164}
\bpages{1139--1160}.
\bid{issn={0016-6731}, pmcid={1462617}, pmid={12871921}}
\end{barticle}
%
\bptok{imsref}%
\endbibitem

\bibitem[\protect\citeauthoryear{Beaumont, Zhang and
Balding}{2002}]{beaumont2002approximate}
%
\begin{barticle}[auto:parserefs-M02]
\bauthor{\bsnm{Beaumont},~\bfnm{M.~A.}\binits{M.~A.}},
\bauthor{\bsnm{Zhang},~\bfnm{W.}\binits{W.}} \AND
\bauthor{\bsnm{Balding},~\bfnm{D.~J.}\binits{D.~J.}}
(\byear{2002}).
\btitle{Approximate Bayesian computation in population genetics}.
\bjournal{Genetics}
\bvolume{162}
\bpages{2025--2035}.
\end{barticle}
%
\bptok{imsref}%
\endbibitem

\bibitem[\protect\citeauthoryear{Besag}{1974}]{besag1974spatial}
%
\begin{barticle}[mr]
\bauthor{\bsnm{Besag},~\bfnm{Julian}\binits{J.}}
(\byear{1974}).
\btitle{Spatial interaction and the statistical analysis of lattice systems}.
\bjournal{J. Roy. Statist. Soc. Ser. B}
\bvolume{36}
\bpages{192--236}.
\bid{issn={0035-9246}, mr={0373208}}
\bptnote{check related}%
\end{barticle}
%
\bptok{imsref}%
\endbibitem

\bibitem[\protect\citeauthoryear{Besag}{1986}]{besag1986statistical}
%
\begin{barticle}[mr]
\bauthor{\bsnm{Besag},~\bfnm{Julian}\binits{J.}}
(\byear{1986}).
\btitle{On the statistical analysis of dirty pictures}.
\bjournal{J.~Roy. Statist. Soc. Ser. B}
\bvolume{48}
\bpages{259--302}.
\bid{issn={0035-9246}, mr={0876840}}
\bptnote{check volume}%
\end{barticle}
%
\bptok{imsref}%
\endbibitem

\bibitem[\protect\citeauthoryear{Besag and Moran}{1975}]{besag1975estimation}
%
\begin{barticle}[mr]
\bauthor{\bsnm{Besag},~\bfnm{J.~E.}\binits{J.~E.}} \AND
\bauthor{\bsnm{Moran},~\bfnm{P.~A.~P.}\binits{P.~A.~P.}}
(\byear{1975}).
\btitle{On the estimation and testing of spatial interaction in
{G}aussian lattice processes}.
\bjournal{Biometrika}
\bvolume{62}
\bpages{555--562}.
\bid{issn={0006-3444}, mr={0391451}}
\end{barticle}
%
\bptok{imsref}%
\endbibitem

\bibitem[\protect\citeauthoryear{Beskos et~al.}{2006}]{Beskos2006id}
%
\begin{barticle}[mr]
\bauthor{\bsnm{Beskos},~\bfnm{Alexandros}\binits{A.}},
\bauthor{\bsnm{Papaspiliopoulos},~\bfnm{Omiros}\binits{O.}},
\bauthor{\bsnm{Roberts},~\bfnm{Gareth~O.}\binits{G.~O.}} \AND
\bauthor{\bsnm{Fearnhead},~\bfnm{Paul}\binits{P.}}
(\byear{2006}).
\btitle{Exact and computationally efficient likelihood-based
estimation for discretely observed diffusion processes}.
\bjournal{J. R. Stat. Soc. Ser. B. Stat. Methodol.}
\bvolume{68}
\bpages{333--382}.
\bid{doi={10.1111/j.1467-9868.2006.00552.x}, issn={1369-7412}, mr={2278331}}
\bptnote{check related}%
\end{barticle}
%
\bptok{imsref}%
\endbibitem

\bibitem[\protect\citeauthoryear{Bhanot and
Kennedy}{1985}]{bhanot1985bosonic}
%
\begin{barticle}[auto:parserefs-M02]
\bauthor{\bsnm{Bhanot},~\bfnm{G.}\binits{G.}} \AND
\bauthor{\bsnm{Kennedy},~\bfnm{A.}\binits{A.}}
(\byear{1985}).
\btitle{Bosonic lattice gauge theory with noise}.
\bjournal{Phys. Lett. B}
\bvolume{157}
\bpages{70--76}.
\end{barticle}
%
\bptok{imsref}%
\endbibitem

\bibitem[\protect\citeauthoryear{Bolin and Lindgren}{2011}]{bolin2009}
%
\begin{barticle}[mr]
\bauthor{\bsnm{Bolin},~\bfnm{David}\binits{D.}} \AND
\bauthor{\bsnm{Lindgren},~\bfnm{Finn}\binits{F.}}
(\byear{2011}).
\btitle{Spatial models generated by nested stochastic partial
differential equations, with an application to global ozone mapping}.
\bjournal{Ann. Appl. Stat.}
\bvolume{5}
\bpages{523--550}.
\bid{doi={10.1214/10-AOAS383}, issn={1932-6157}, mr={2810408}}
\end{barticle}
%
\bptok{imsref}%
\endbibitem

\bibitem[\protect\citeauthoryear{Booth}{2007}]{booth2007unbiased}
%
\begin{barticle}[auto:parserefs-M02]
\bauthor{\bsnm{Booth},~\bfnm{T.}\binits{T.}}
(\byear{2007}).
\btitle{Unbiased Monte Carlo estimation of the reciprocal of an integral}.
\bjournal{Nucl. Sci. Eng.}
\bvolume{156}
\bpages{403--407}.
\end{barticle}
%
\bptok{imsref}%
\endbibitem

\bibitem[\protect\citeauthoryear{Caimo and Friel}{2011}]{caimo2011bayesian}
%
\begin{barticle}[auto:parserefs-M02]
\bauthor{\bsnm{Caimo},~\bfnm{A.}\binits{A.}} \AND
\bauthor{\bsnm{Friel},~\bfnm{N.}\binits{N.}}
(\byear{2011}).
\btitle{Bayesian inference for exponential random graph models}.
\bjournal{Soc. Netw.}
\bvolume{33}
\bpages{41--55}.
\end{barticle}
%
\bptok{imsref}%
\endbibitem

\bibitem[\protect\citeauthoryear{Carter and
Cashwell}{1975}]{carter1975particle}
%
\begin{bmisc}[mr]
\bauthor{\bsnm{Carter},~\bfnm{L.~L.}\binits{L.~L.}} \AND
\bauthor{\bsnm{Cashwell},~\bfnm{E.~D.}\binits{E.~D.}}
(\byear{1975}).
\bhowpublished{Particle-transport simulation with the {M}onte {C}arlo method.
Technical report, Los Alamos Scientific Lab., N. Mex. (USA).}
\bid{mr={0416421}}
\end{bmisc}
%
\bptok{imsref}%
\endbibitem

\bibitem[\protect\citeauthoryear{Cressie and
Johannesson}{2008}]{Cressie2008Fixed}
%
\begin{barticle}[mr]
\bauthor{\bsnm{Cressie},~\bfnm{Noel}\binits{N.}} \AND
\bauthor{\bsnm{Johannesson},~\bfnm{Gardar}\binits{G.}}
(\byear{2008}).
\btitle{Fixed rank kriging for very large spatial data sets}.
\bjournal{J. R. Stat. Soc. Ser. B. Stat. Methodol.}
\bvolume{70}
\bpages{209--226}.
\bid{doi={10.1111/j.1467-9868.2007.00633.x}, issn={1369-7412}, mr={2412639}}
\end{barticle}
%
\bptok{imsref}%
\endbibitem

\bibitem[\protect\citeauthoryear{Del~Moral, Doucet and
Jasra}{2006}]{delmoral2006}
%
\begin{barticle}[mr]
\bauthor{\bsnm{Del Moral},~\bfnm{Pierre}\binits{P.}},
\bauthor{\bsnm{Doucet},~\bfnm{Arnaud}\binits{A.}} \AND
\bauthor{\bsnm{Jasra},~\bfnm{Ajay}\binits{A.}}
(\byear{2006}).
\btitle{Sequential {M}onte {C}arlo samplers}.
\bjournal{J. R. Stat. Soc. Ser. B. Stat. Methodol.}
\bvolume{68}
\bpages{411--436}.
\bid{doi={10.1111/j.1467-9868.2006.00553.x}, issn={1369-7412}, mr={2278333}}
\end{barticle}
%
\bptok{imsref}%
\endbibitem

\bibitem[\protect\citeauthoryear{Diggle}{1990}]{diggle1990point}
%
\begin{barticle}[auto:parserefs-M02]
\bauthor{\bsnm{Diggle},~\bfnm{P.~J.}\binits{P.~J.}}
(\byear{1990}).
\btitle{A point process modelling approach to raised incidence of a
rare phenomenon in the vicinity of a prespecified point}.
\bjournal{J. Roy. Statist. Soc. Ser. A}
\bpages{349--362}.
\end{barticle}
%
\bptok{imsref}%
\endbibitem

\bibitem[\protect\citeauthoryear{Douc and Robert}{2011}]{douc2011vanilla}
%
\begin{barticle}[mr]
\bauthor{\bsnm{Douc},~\bfnm{Randal}\binits{R.}} \AND
\bauthor{\bsnm{Robert},~\bfnm{Christian~P.}\binits{C.~P.}}
(\byear{2011}).
\btitle{A vanilla {R}ao--{B}lackwellization of
{M}etropolis--{H}astings algorithms}.
\bjournal{Ann. Statist.}
\bvolume{39}
\bpages{261--277}.
\bid{doi={10.1214/10-AOS838}, issn={0090-5364}, mr={2797846}}
\end{barticle}
%
\bptok{imsref}%
\endbibitem

\bibitem[\protect\citeauthoryear{Doucet, Pitt and
Kohn}{2012}]{doucet2012efficient}
%
\begin{bmisc}[auto:parserefs-M02]
\bauthor{\bsnm{Doucet},~\bfnm{A.}\binits{A.}},
\bauthor{\bsnm{Pitt},~\bfnm{M.}\binits{M.}} \AND
\bauthor{\bsnm{Kohn},~\bfnm{R.}\binits{R.}}
(\byear{2012}).
\bhowpublished{Efficient implementation of Markov chain Monte Carlo
when using an unbiased likelihood estimator.
Preprint. Available at \arxivurl{arXiv:1210.1871}.}
\end{bmisc}
%
\bptok{imsref}%
\endbibitem

\bibitem[\protect\citeauthoryear{Eidsvik
et~al.}{2014}]{eidsvik2013estimation}
%
\begin{barticle}[mr]
\bauthor{\bsnm{Eidsvik},~\bfnm{Jo}\binits{J.}},
\bauthor{\bsnm{Shaby},~\bfnm{Benjamin~A.}\binits{B.~A.}},
\bauthor{\bsnm{Reich},~\bfnm{Brian~J.}\binits{B.~J.}},
\bauthor{\bsnm{Wheeler},~\bfnm{Matthew}\binits{M.}} \AND
\bauthor{\bsnm{Niemi},~\bfnm{Jarad}\binits{J.}}
(\byear{2014}).
\btitle{Estimation and prediction in spatial models with block
composite likelihoods}.
\bjournal{J. Comput. Graph. Statist.}
\bvolume{23}
\bpages{295--315}.
\bid{doi={10.1080/10618600.2012.760460}, issn={1061-8600}, mr={3215812}}
\bptnote{check year}%
\end{barticle}
%
\bptok{imsref}%
\endbibitem

\bibitem[\protect\citeauthoryear{Everitt}{2012}]{everitt2012bayesian}
%
\begin{barticle}[mr]
\bauthor{\bsnm{Everitt},~\bfnm{Richard~G.}\binits{R.~G.}}
(\byear{2012}).
\btitle{Bayesian parameter estimation for latent {M}arkov random
fields and social networks}.
\bjournal{J. Comput. Graph. Statist.}
\bvolume{21}
\bpages{940--960}.
\bid{doi={10.1080/10618600.2012.687493}, issn={1061-8600}, mr={3005805}}
\bptnote{check volume, check pages, check year}%
\end{barticle}
%
\bptok{imsref}%
\endbibitem

\bibitem[\protect\citeauthoryear{Fearnhead, Papaspiliopoulos and Roberts}{2008}]{Fearnhead2008kk}
%
\begin{barticle}[mr]
\bauthor{\bsnm{Fearnhead},~\bfnm{Paul}\binits{P.}},
\bauthor{\bsnm{Papaspiliopoulos},~\bfnm{Omiros}\binits{O.}} \AND
\bauthor{\bsnm{Roberts},~\bfnm{Gareth~O.}\binits{G.~O.}}
(\byear{2008}).
\btitle{Particle filters for partially observed diffusions}.
\bjournal{J. R. Stat. Soc. Ser. B. Stat. Methodol.}
\bvolume{70}
\bpages{755--777}.
\bid{doi={10.1111/j.1467-9868.2008.00661.x}, issn={1369-7412}, mr={2523903}}
\end{barticle}
%
\bptok{imsref}%
\endbibitem

\bibitem[\protect\citeauthoryear{Friel and
Pettitt}{2004}]{friel2004likelihood}
%
\begin{barticle}[mr]
\bauthor{\bsnm{Friel},~\bfnm{N.}\binits{N.}} \AND
\bauthor{\bsnm{Pettitt},~\bfnm{A.~N.}\binits{A.~N.}}
(\byear{2004}).
\btitle{Likelihood estimation and inference for the autologistic model}.
\bjournal{J. Comput. Graph. Statist.}
\bvolume{13}
\bpages{232--246}.
\bid{doi={10.1198/1061860043029}, issn={1061-8600}, mr={2044879}}
\bptnote{check pages}%
\end{barticle}
%
\bptok{imsref}%
\endbibitem

\bibitem[\protect\citeauthoryear{Friel et~al.}{2009}]{friel2009bayesian}
%
\begin{barticle}[mr]
\bauthor{\bsnm{Friel},~\bfnm{N.}\binits{N.}},
\bauthor{\bsnm{Pettitt},~\bfnm{A.~N.}\binits{A.~N.}},
\bauthor{\bsnm{Reeves},~\bfnm{R.}\binits{R.}} \AND
\bauthor{\bsnm{Wit},~\bfnm{E.}\binits{E.}}
(\byear{2009}).
\btitle{Bayesian inference in hidden {M}arkov random fields for binary
data defined on large lattices}.
\bjournal{J. Comput. Graph. Statist.}
\bvolume{18}
\bpages{243--261}.
\bid{doi={10.1198/jcgs.2009.06148}, issn={1061-8600}, mr={2749831}}
\end{barticle}
%
\bptok{imsref}%
\endbibitem

\bibitem[\protect\citeauthoryear{Gelman and
Meng}{1998}]{gelman1998simulating}
%
\begin{barticle}[mr]
\bauthor{\bsnm{Gelman},~\bfnm{Andrew}\binits{A.}} \AND
\bauthor{\bsnm{Meng},~\bfnm{Xiao-Li}\binits{X.-L.}}
(\byear{1998}).
\btitle{Simulating normalizing constants: From importance sampling to
bridge sampling to path sampling}.
\bjournal{Statist. Sci.}
\bvolume{13}
\bpages{163--185}.
\bid{doi={10.1214/ss/1028905934}, issn={0883-4237}, mr={1647507}}
\bptnote{check volume}%
\end{barticle}
%
\bptok{imsref}%
\endbibitem

\bibitem[\protect\citeauthoryear{Gelman et~al.}{1995}]{Gelman03}
%
\begin{bbook}[mr]
\bauthor{\bsnm{Gelman},~\bfnm{Andrew}\binits{A.}},
\bauthor{\bsnm{Carlin},~\bfnm{John~B.}\binits{J.~B.}},
\bauthor{\bsnm{Stern},~\bfnm{Hal~S.}\binits{H.~S.}} \AND
\bauthor{\bsnm{Rubin},~\bfnm{Donald~B.}\binits{D.~B.}}
(\byear{1995}).
\btitle{Bayesian Data Analysis}.
\bpublisher{Chapman \& Hall},
\blocation{London}.
\bid{mr={1385925}}
\bptnote{check year}%
\end{bbook}
%
\bptok{imsref}%
\endbibitem

\bibitem[\protect\citeauthoryear{Ghaoui and Gueye}{2009}]{el2008convex}
%
\begin{bincollection}[auto:parserefs-M02]
\bauthor{\bsnm{Ghaoui},~\bfnm{L.~E.}\binits{L.~E.}} \AND
\bauthor{\bsnm{Gueye},~\bfnm{A.}\binits{A.}}
(\byear{2009}).
\btitle{A convex upper bound on the log-partition function for binary
distributions}.
In \bbooktitle{Advances in Neural Information Processing Systems}
(\beditor{\bfnm{D.}\binits{D.}~\bsnm{Koller}},
\beditor{\bfnm{D.}\binits{D.}~\bsnm{Schuurmans}},
\beditor{\bfnm{Y.}\binits{Y.}~\bsnm{Bengio}} \AND
\beditor{\bfnm{L.}\binits{L.}~\bsnm{Bottou}}, eds.)
\bvolume{21}
\bpages{409--416}.
\bpublisher{Neural Information Processing Systems (NIPS)}.
\end{bincollection}
%
\bptok{imsref}%
\endbibitem

\bibitem[\protect\citeauthoryear{Gilks}{1996}]{Gilks99}
%
\begin{bbook}[mr]
\bauthor{\bsnm{Gilks},~\bfnm{W.~R.}\binits{W.~R.}}
(\byear{1996}).
\btitle{Markov Chain {M}onte {C}arlo in Practice}.
\bpublisher{Chapman \& Hall},
\blocation{London}.
\bid{doi={10.1007/978-1-4899-4485-6}, mr={1397966}}
\bptnote{check year}%
\end{bbook}
%
\bptok{imsref}%
\endbibitem

\bibitem[\protect\citeauthoryear{Glynn and Rhee}{2014}]{glynn2014exact}
%
\begin{barticle}[mr]
\bauthor{\bsnm{Glynn},~\bfnm{Peter~W.}\binits{P.~W.}} \AND
\bauthor{\bsnm{Rhee},~\bfnm{Chang-Han}\binits{C.-H.}}
(\byear{2014}).
\btitle{Exact estimation for {M}arkov chain equilibrium expectations}.
\bjournal{J. Appl. Probab.}
\bvolume{51A}
\bpages{377--389}.
\bid{doi={10.1239/jap/1417528487}, issn={0021-9002}, mr={3317370}}
\bptnote{check volume}%
\end{barticle}
%
\bptok{imsref}%
\endbibitem

\bibitem[\protect\citeauthoryear{Golub and Van~Loan}{1996}]{Golub1996}
%
\begin{bbook}[mr]
\bauthor{\bsnm{Golub},~\bfnm{Gene~H.}\binits{G.~H.}} \AND
\bauthor{\bsnm{Van Loan},~\bfnm{Charles~F.}\binits{C.~F.}}
(\byear{1996}).
\btitle{Matrix Computations},
\bedition{3rd} ed.
\bpublisher{Johns Hopkins Univ. Press},
\blocation{Baltimore, MD}.
\bid{mr={1417720}}
\end{bbook}
%
\bptok{imsref}%
\endbibitem

\bibitem[\protect\citeauthoryear{Goodreau, Kitts and
Morris}{2009}]{goodreau2009birds}
%
\begin{barticle}[pbm]
\bauthor{\bsnm{Goodreau},~\bfnm{Steven~M.}\binits{S.~M.}},
\bauthor{\bsnm{Kitts},~\bfnm{James~A.}\binits{J.~A.}} \AND
\bauthor{\bsnm{Morris},~\bfnm{Martina}\binits{M.}}
(\byear{2009}).
\btitle{Birds of a feather, or friend of a friend? Using exponential
random graph models to investigate adolescent social networks}.
\bjournal{Demography}
\bvolume{46}
\bpages{103--125}.
\bid{issn={0070-3370}, pmcid={2831261}, pmid={19348111}}
\end{barticle}
%
\bptok{imsref}%
\endbibitem

\bibitem[\protect\citeauthoryear{Green and
Richardson}{2002}]{green2002hidden}
%
\begin{barticle}[mr]
\bauthor{\bsnm{Green},~\bfnm{Peter~J.}\binits{P.~J.}} \AND
\bauthor{\bsnm{Richardson},~\bfnm{Sylvia}\binits{S.}}
(\byear{2002}).
\btitle{Hidden {M}arkov models and disease mapping}.
\bjournal{J. Amer. Statist. Assoc.}
\bvolume{97}
\bpages{1055--1070}.
\bid{doi={10.1198/016214502388618870}, issn={0162-1459}, mr={1951259}}
\end{barticle}
%
\bptok{imsref}%
\endbibitem

\bibitem[\protect\citeauthoryear{Grelaud, Robert and
Marin}{2009}]{grelaud2009abc}
%
\begin{barticle}[mr]
\bauthor{\bsnm{Grelaud},~\bfnm{Aude}\binits{A.}},
\bauthor{\bsnm{Robert},~\bfnm{Christian~P.}\binits{C.~P.}} \AND
\bauthor{\bsnm{Marin},~\bfnm{Jean-Michel}\binits{J.-M.}}
(\byear{2009}).
\btitle{A{BC} methods for model choice in {G}ibbs random fields}.
\bjournal{C. R. Math. Acad. Sci. Paris}
\bvolume{347}
\bpages{205--210}.
\bid{doi={10.1016/j.crma.2008.12.009}, issn={1631-073X}, mr={2538114}}
\end{barticle}
%
\bptok{imsref}%
\endbibitem

\bibitem[\protect\citeauthoryear{Gu and Zhu}{2001}]{gu2001maximum}
%
\begin{barticle}[mr]
\bauthor{\bsnm{Gu},~\bfnm{Ming~Gao}\binits{M.~G.}} \AND
\bauthor{\bsnm{Zhu},~\bfnm{Hong-Tu}\binits{H.-T.}}
(\byear{2001}).
\btitle{Maximum likelihood estimation for spatial models by {M}arkov
chain {M}onte {C}arlo stochastic approximation}.
\bjournal{J. R. Stat. Soc. Ser. B. Stat. Methodol.}
\bvolume{63}
\bpages{339--355}.
\bid{doi={10.1111/1467-9868.00289}, issn={1369-7412}, mr={1841419}}
\end{barticle}
%
\bptok{imsref}%
\endbibitem

\bibitem[\protect\citeauthoryear{Heikkinen and Hogmander}{1994}]{heikkinen1994fully}
%
\begin{barticle}[auto:parserefs-M02]
\bauthor{\bsnm{Heikkinen},~\bfnm{J.}\binits{J.}} \AND
\bauthor{\bsnm{Hogmander},~\bfnm{H.}\binits{H.}}
(\byear{1994}).
\btitle{Fully Bayesian approach to image restoration with an
application in biogeography}.
\bjournal{Applied Statistics}
\bvolume{43}
\bpages{569--582}.
\end{barticle}
%
\bptok{imsref}%
\endbibitem

\bibitem[\protect\citeauthoryear{Hendricks and
Booth}{1985}]{hendricks1985mcnp}
%
\begin{binproceedings}[auto]
\bauthor{\bsnm{Hendricks},~\bfnm{J.}\binits{J.}} \AND
\bauthor{\bsnm{Booth},~\bfnm{T.}\binits{T.}}
(\byear{1985}).
\btitle{Mcnp variance reduction overview}.
In \bbooktitle{Monte-Carlo Methods and Applications in Neutronics,
Photonics and Statistical Physics}
\bpages{83--92}.
\bpublisher{Springer},
\blocation{Berlin}.
\end{binproceedings}
%
\bptok{imsref}%
\endbibitem

\bibitem[\protect\citeauthoryear{Hughes, Haran and
Caragea}{2011}]{hughes2011autologistic}
%
\begin{barticle}[mr]
\bauthor{\bsnm{Hughes},~\bfnm{John}\binits{J.}},
\bauthor{\bsnm{Haran},~\bfnm{Murali}\binits{M.}} \AND
\bauthor{\bsnm{Caragea},~\bfnm{Petru{\c{t}}a~C.}\binits{P.~C.}}
(\byear{2011}).
\btitle{Autologistic models for binary data on a lattice}.
\bjournal{Environmetrics}
\bvolume{22}
\bpages{857--871}.
\bid{doi={10.1002/env.1102}, issn={1180-4009}, mr={2861051}}
\end{barticle}
%
\bptok{imsref}%
\endbibitem

\bibitem[\protect\citeauthoryear{Illian et~al.}{2012}]{illian2010fitting}
%
\begin{barticle}[auto:parserefs-M02]
\bauthor{\bsnm{Illian},~\bfnm{J.}\binits{J.}},
\bauthor{\bsnm{S{\o}rbye},~\bfnm{S.}\binits{S.}},
\bauthor{\bsnm{Rue},~\bfnm{H.}\binits{H.}} \AND
\bauthor{\bsnm{Hendrichsen},~\bfnm{D.}\binits{D.}}
(\byear{2012}).
\btitle{Using INLA to fit a complex point process model with temporally varying effects--a case study}.
\bjournal{J. Environ. Statist.}
\bvolume{3}
\bpages{1--25}.
\end{barticle}
%
\bptok{imsref}%
\endbibitem

\bibitem[\protect\citeauthoryear{Ising}{1925}]{ising1925beitrag}
%
\begin{barticle}[auto:parserefs-M02]
\bauthor{\bsnm{Ising},~\bfnm{E.}\binits{E.}}
(\byear{1925}).
\btitle{Beitrag zur Theorie des Ferromagnetismus}.
\bjournal{Zeitschrift F\"{u}r Physik A Hadrons and Nuclei}
\bvolume{31}
\bpages{253--258}.
\end{barticle}
%
\bptok{imsref}%
\endbibitem

\bibitem[\protect\citeauthoryear{Jacob and Thiery}{2013}]{jacob2013non}
%
\begin{bmisc}[auto:parserefs-M02]
\bauthor{\bsnm{Jacob},~\bfnm{P.~E.}\binits{P.~E.}} \AND
\bauthor{\bsnm{Thiery},~\bfnm{A.~H.}\binits{A.~H.}}
(\byear{2013}).
\bhowpublished{On non-negative unbiased estimators.
Preprint. Available at \arxivurl{arXiv:1309.6473}.}
\end{bmisc}
%
\bptok{imsref}%
\endbibitem

\bibitem[\protect\citeauthoryear{Jin and Liang}{2014}]{jin2014use}
%
\begin{barticle}[mr]
\bauthor{\bsnm{Jin},~\bfnm{Ick~Hoon}\binits{I.~H.}} \AND
\bauthor{\bsnm{Liang},~\bfnm{Faming}\binits{F.}}
(\byear{2014}).
\btitle{Use of SAMC for {B}ayesian analysis of statistical models with
intractable normalizing constants}.
\bjournal{Comput. Statist. Data Anal.}
\bvolume{71}
\bpages{402--416}.
\bid{doi={10.1016/j.csda.2012.07.005}, issn={0167-9473}, mr={3131979}}
\end{barticle}
%
\bptok{imsref}%
\endbibitem

\bibitem[\protect\citeauthoryear{Joo, Horvath and
Liu}{2003}]{joo2003kentucky}
%
\begin{barticle}[auto:parserefs-M02]
\bauthor{\bsnm{Joo},~\bfnm{B.}\binits{B.}},
\bauthor{\bsnm{Horvath},~\bfnm{I.}\binits{I.}} \AND
\bauthor{\bsnm{Liu},~\bfnm{K.}\binits{K.}}
(\byear{2003}).
\btitle{The Kentucky noisy Monte Carlo algorithm for Wilson dynamical
fermions}.
\bjournal{Phys. Rev. D}
\bvolume{67}
\bpages{074505}.
\end{barticle}
%
\bptok{imsref}%
\endbibitem

\bibitem[\protect\citeauthoryear{Jun and Stein}{2008}]{Jun2008}
%
\begin{barticle}[mr]
\bauthor{\bsnm{Jun},~\bfnm{Mikyoung}\binits{M.}} \AND
\bauthor{\bsnm{Stein},~\bfnm{Michael~L.}\binits{M.~L.}}
(\byear{2008}).
\btitle{Nonstationary covariance models for global data}.
\bjournal{Ann. Appl. Stat.}
\bvolume{2}
\bpages{1271--1289}.
\bid{doi={10.1214/08-AOAS183}, issn={1932-6157}, mr={2655659}}
\end{barticle}
%
\bptok{imsref}%
\endbibitem

\bibitem[\protect\citeauthoryear{Kendall}{2005}]{kendall2005notes}
%
\begin{bmisc}[automr]
\bauthor{\bsnm{Kendall},~\bfnm{W.~S.}\binits{W.~S.}}
(\byear{2005}).
\bhowpublished{Notes on perfect simulation. \textit{Markov Chain
Monte Carlo}: \textit{Innovations and Applications} \textbf{7}. World Scientific, Singapore}.
\end{bmisc}
%
\bptok{imsref}%
\endbibitem

\bibitem[\protect\citeauthoryear{Kennedy and Kuti}{1985}]{kennedy1985noise}
%
\begin{barticle}[auto:parserefs-M02]
\bauthor{\bsnm{Kennedy},~\bfnm{A.}\binits{A.}} \AND
\bauthor{\bsnm{Kuti},~\bfnm{J.}\binits{J.}}
(\byear{1985}).
\btitle{Noise without noise: A new Monte Carlo method}.
\bjournal{Phys. Rev. Lett.}
\bvolume{54}
\bpages{2473--2476}.
\end{barticle}
%
\bptok{imsref}%
\endbibitem

\bibitem[\protect\citeauthoryear{Kent}{1982}]{kent1982}
%
\begin{barticle}[mr]
\bauthor{\bsnm{Kent},~\bfnm{John~T.}\binits{J.~T.}}
(\byear{1982}).
\btitle{The {F}isher--{B}ingham distribution on the sphere}.
\bjournal{J. Roy. Statist. Soc. Ser. B}
\bvolume{44}
\bpages{71--80}.
\bid{issn={0035-9246}, mr={0655376}}
\end{barticle}
%
\bptok{imsref}%
\endbibitem

\bibitem[\protect\citeauthoryear{Korattikara, Chen and Welling}{2014}]{korattikara2014austerity}
%
\begin{binproceedings}[auto:parserefs-M02]
\bauthor{\bsnm{Korattikara},~\bfnm{A.}\binits{A.}},
\bauthor{\bsnm{Chen},~\bfnm{Y.}\binits{Y.}} \AND
\bauthor{\bsnm{Welling},~\bfnm{M.}\binits{M.}}
(\byear{2014}).
\btitle{Austerity in MCMC land: Cutting the Metropolis--Hastings budget}.
In \bbooktitle{Proceedings of the 31st International Conference on
Machine Learning}
\bpages{181--189}.
\bpublisher{JMLR Workshop and Conference Proceedings}.
\end{binproceedings}
%
\bptok{imsref}%
\endbibitem

\bibitem[\protect\citeauthoryear{Liang}{2010}]{liang2010double}
%
\begin{barticle}[mr]
\bauthor{\bsnm{Liang},~\bfnm{Faming}\binits{F.}}
(\byear{2010}).
\btitle{A double {M}etropolis--{H}astings sampler for spatial models
with intractable normalizing constants}.
\bjournal{J. Stat. Comput. Simul.}
\bvolume{80}
\bpages{1007--1022}.
\bid{doi={10.1080/00949650902882162}, issn={0094-9655}, mr={2742519}}
\end{barticle}
%
\bptok{imsref}%
\endbibitem

\bibitem[\protect\citeauthoryear{Liang, Liu and
Carroll}{2007}]{liang2007stochastic}
%
\begin{barticle}[mr]
\bauthor{\bsnm{Liang},~\bfnm{Faming}\binits{F.}},
\bauthor{\bsnm{Liu},~\bfnm{Chuanhai}\binits{C.}} \AND
\bauthor{\bsnm{Carroll},~\bfnm{Raymond~J.}\binits{R.~J.}}
(\byear{2007}).
\btitle{Stochastic approximation in {M}onte {C}arlo computation}.
\bjournal{J. Amer. Statist. Assoc.}
\bvolume{102}
\bpages{305--320}.
\bid{doi={10.1198/016214506000001202}, issn={0162-1459}, mr={2345544}}
\end{barticle}
%
\bptok{imsref}%
\endbibitem

\bibitem[\protect\citeauthoryear{Liechty, Liechty and M{\"
u}ller}{2009}]{liechty2009shadow}
%
\begin{barticle}[mr]
\bauthor{\bsnm{Liechty},~\bfnm{Merrill~W.}\binits{M.~W.}},
\bauthor{\bsnm{Liechty},~\bfnm{John~C.}\binits{J.~C.}} \AND
\bauthor{\bsnm{M{\"u}ller},~\bfnm{Peter}\binits{P.}}
(\byear{2009}).
\btitle{The shadow prior}.
\bjournal{J. Comput. Graph. Statist.}
\bvolume{18}
\bpages{368--383}.
\bid{doi={10.1198/jcgs.2009.07072}, issn={1061-8600}, mr={2749837}}
\end{barticle}
%
\bptok{imsref}%
\endbibitem

\bibitem[\protect\citeauthoryear{Lin, Liu and Sloan}{2000}]{lin2000noisy}
%
\begin{barticle}[auto:parserefs-M02]
\bauthor{\bsnm{Lin},~\bfnm{L.}\binits{L.}},
\bauthor{\bsnm{Liu},~\bfnm{K.}\binits{K.}} \AND
\bauthor{\bsnm{Sloan},~\bfnm{J.}\binits{J.}}
(\byear{2000}).
\btitle{A noisy Monte Carlo algorithm}.
\bjournal{Phys. Rev. D}
\bvolume{61}
\bpages{074505}.
\end{barticle}
%
\bptok{imsref}%
\endbibitem

\bibitem[\protect\citeauthoryear{Lindgren, Rue and Lindstr{\"
o}m}{2011}]{Lindgren2011}
%
\begin{barticle}[mr]
\bauthor{\bsnm{Lindgren},~\bfnm{Finn}\binits{F.}},
\bauthor{\bsnm{Rue},~\bfnm{H{\aa}vard}\binits{H.}} \AND
\bauthor{\bsnm{Lindstr{\"o}m},~\bfnm{Johan}\binits{J.}}
(\byear{2011}).
\btitle{An explicit link between {G}aussian fields and {G}aussian
{M}arkov random fields: The stochastic partial differential equation approach}.
\bjournal{J.~R. Stat. Soc. Ser. B. Stat. Methodol.}
\bvolume{73}
\bpages{423--498}.
\bid{doi={10.1111/j.1467-9868.2011.00777.x}, issn={1369-7412}, mr={2853727}}
\bptnote{check related}%
\end{barticle}
%
\bptok{imsref}%
\endbibitem

\bibitem[\protect\citeauthoryear{Liu}{2001}]{Liu01}
%
\begin{bbook}[mr]
\bauthor{\bsnm{Liu},~\bfnm{Jun~S.}\binits{J.~S.}}
(\byear{2001}).
\btitle{Monte {C}arlo Strategies in Scientific Computing}.
\bpublisher{Springer},
\blocation{New York}.
\bid{mr={1842342}}
\end{bbook}
%
\bptok{imsref}%
\endbibitem

\bibitem[\protect\citeauthoryear{Lux and Koblinger}{1991}]{lux1991monte}
%
\begin{bbook}[auto:parserefs-M02]
\bauthor{\bsnm{Lux},~\bfnm{I.}\binits{I.}} \AND
\bauthor{\bsnm{Koblinger},~\bfnm{L.}\binits{L.}}
(\byear{1991}).
\btitle{Monte Carlo Particle Transport Methods: Neutron and Photon
Calculations, Vol.~102}.
\bpublisher{CRC press},
\blocation{Boca Raton}.
\end{bbook}
%
\bptok{imsref}%
\endbibitem

\bibitem[\protect\citeauthoryear{MacKay}{2003}]{mackay2003information}
%
\begin{bbook}[mr]
\bauthor{\bsnm{MacKay},~\bfnm{David~J.~C.}\binits{D.~J.~C.}}
(\byear{2003}).
\btitle{Information Theory, Inference and Learning Algorithms}.
\bpublisher{Cambridge Univ. Press},
\blocation{New York}.
\bid{mr={2012999}}
\end{bbook}
%
\bptok{imsref}%
\endbibitem

\bibitem[\protect\citeauthoryear{Marin et~al.}{2012}]{marinabc2012}
%
\begin{barticle}[mr]
\bauthor{\bsnm{Marin},~\bfnm{Jean-Michel}\binits{J.-M.}},
\bauthor{\bsnm{Pudlo},~\bfnm{Pierre}\binits{P.}},
\bauthor{\bsnm{Robert},~\bfnm{Christian~P.}\binits{C.~P.}} \AND
\bauthor{\bsnm{Ryder},~\bfnm{Robin~J.}\binits{R.~J.}}
(\byear{2012}).
\btitle{Approximate {B}ayesian computational methods}.
\bjournal{Stat. Comput.}
\bvolume{22}
\bpages{1167--1180}.
\bid{doi={10.1007/s11222-011-9288-2}, issn={0960-3174}, mr={2992292}}
\bptnote{check volume, check pages, check year}%
\end{barticle}
%
\bptok{imsref}%
\endbibitem

\bibitem[\protect\citeauthoryear{McLeish}{2011}]{mcleish2011general}
%
\begin{barticle}[mr]
\bauthor{\bsnm{McLeish},~\bfnm{Don}\binits{D.}}
(\byear{2011}).
\btitle{A general method for debiasing a {M}onte {C}arlo estimator}.
\bjournal{Monte Carlo Methods Appl.}
\bvolume{17}
\bpages{301--315}.
\bid{doi={10.1515/mcma.2011.013}, issn={0929-9629}, mr={2890424}}
\end{barticle}
%
\bptok{imsref}%
\endbibitem

\bibitem[\protect\citeauthoryear{M{\o}ller and Waagepetersen}{2004}]{moller2003statistical}
%
\begin{bbook}[mr]
\bauthor{\bsnm{M{\o}ller},~\bfnm{Jesper}\binits{J.}} \AND
\bauthor{\bsnm{Waagepetersen},~\bfnm{Rasmus~Plenge}\binits{R.~P.}}
(\byear{2004}).
\btitle{Statistical Inference and Simulation for Spatial Point Processes}.
\bpublisher{Chapman \& Hall/CRC},
\blocation{Boca Raton, FL}.
\bid{mr={2004226}}
\bptnote{check year}%
\end{bbook}
%
\bptok{imsref}%
\endbibitem

\bibitem[\protect\citeauthoryear{M{\o}ller
et~al.}{2006}]{moller2006efficient}
%
\begin{barticle}[mr]
\bauthor{\bsnm{M{\o}ller},~\bfnm{J.}\binits{J.}},
\bauthor{\bsnm{Pettitt},~\bfnm{A.~N.}\binits{A.~N.}},
\bauthor{\bsnm{Reeves},~\bfnm{R.}\binits{R.}} \AND
\bauthor{\bsnm{Berthelsen},~\bfnm{K.~K.}\binits{K.~K.}}
(\byear{2006}).
\btitle{An efficient {M}arkov chain {M}onte {C}arlo method for
distributions with intractable normalising constants}.
\bjournal{Biometrika}
\bvolume{93}
\bpages{451--458}.
\bid{doi={10.1093/biomet/93.2.451}, issn={0006-3444}, mr={2278096}}
\end{barticle}
%
\bptok{imsref}%
\endbibitem

\bibitem[\protect\citeauthoryear{Moores, Mengersen and Robert}{2014}]{moores2014pre}
%
\begin{bmisc}[auto]
\bauthor{\bsnm{Moores},~\bfnm{M. T.}\binits{M.~T.}},
\bauthor{\bsnm{Mengersen},~\bfnm{K.}\binits{K.}},
\bauthor{\bsnm{Robert},~\bfnm{C. P.}\binits{C.~P.}}
(\byear{2014}).
\bhowpublished{Pre-processing for approximate bayesian computation in image analysis.
Preprint. Available at \arxivurl{arXiv:1403.4359}.}
\end{bmisc}
%
\bptok{imsref}%
\endbibitem

\bibitem[\protect\citeauthoryear{Murray, Ghahramani and MacKay}{2006}]{murray2006mcmc}
%
\begin{binproceedings}[auto:parserefs-M02]
\bauthor{\bsnm{Murray},~\bfnm{I.}\binits{I.}},
\bauthor{\bsnm{Ghahramani},~\bfnm{Z.}\binits{Z.}} \AND
\bauthor{\bsnm{MacKay},~\bfnm{D.}\binits{D.}}
(\byear{2006}).
\btitle{MCMC for doubly-intractable distributions}.
In \bbooktitle{Proceedings of the 22nd Annual Conference on
Uncertainty in Artificial Intelligence (UAI-06)}
\bpages{359--366}.
\bpublisher{AUAI Press}, \blocation{Arlington, VI}.
\end{binproceedings}
%
\bptok{imsref}%
\endbibitem

\bibitem[\protect\citeauthoryear{Neal}{2001}]{Neal98annealedimportance}
%
\begin{barticle}[mr]
\bauthor{\bsnm{Neal},~\bfnm{Radford~M.}\binits{R.~M.}}
(\byear{2001}).
\btitle{Annealed importance sampling}.
\bjournal{Stat. Comput.}
\bvolume{11}
\bpages{125--139}.
\bid{doi={10.1023/A:1008923215028}, issn={0960-3174}, mr={1837132}}
\bptnote{check year}%
\end{barticle}
%
\bptok{imsref}%
\endbibitem

\bibitem[\protect\citeauthoryear{Papaspiliopoulos}{2011}]{papaspil2011}
%
\begin{bincollection}[mr]
\bauthor{\bsnm{Papaspiliopoulos},~\bfnm{Omiros}\binits{O.}}
(\byear{2011}).
\btitle{Monte {C}arlo probabilistic inference for diffusion processes:
A methodological framework}.
In \bbooktitle{Bayesian Time Series Models}
\bpages{82--103}.
\bpublisher{Cambridge Univ. Press},
\blocation{Cambridge}.
\bid{mr={2894234}}
\bptnote{check pages}%
\end{bincollection}
%
\bptok{imsref}%
\endbibitem

\bibitem[\protect\citeauthoryear{Propp and Wilson}{1996}]{propp1996exact}
%
\begin{barticle}[auto]
\bauthor{\bsnm{Propp},~\bfnm{James~Gary}\binits{J.~G.}} \AND
\bauthor{\bsnm{Wilson},~\bfnm{David~Bruce}\binits{D.~B.}}
(\byear{1996}).
\btitle{Exact sampling with coupled Markov chains and applications to
statistical mechanics}.
\bjournal{Random Structures Algorithms}
\bvolume{9}
\bpages{223--252}.
\end{barticle}
%
\bptok{imsref}%
\endbibitem

\bibitem[\protect\citeauthoryear{Rhee and Glynn}{2012}]{rhee2012new}
%
\begin{binproceedings}[auto:parserefs-M02]
\bauthor{\bsnm{Rhee},~\bfnm{C.-H.}\binits{C.-H.}} \AND
\bauthor{\bsnm{Glynn},~\bfnm{P.~W.}\binits{P.~W.}}
(\byear{2012}).
\btitle{A new approach to unbiased estimation for SDE's}.
In \bbooktitle{Proceedings of the Winter Simulation Conference, WSC'12, Berlin, Germany  17:1--17:7}.
\bpublisher{Winter Simulation Conference}.
\end{binproceedings}
%
\bptok{imsref}%
\endbibitem

\bibitem[\protect\citeauthoryear{Robert and
Casella}{2010}]{RRobert+Casella2010}
%
\begin{bbook}[mr]
\bauthor{\bsnm{Robert},~\bfnm{Christian~P.}\binits{C.~P.}} \AND
\bauthor{\bsnm{Casella},~\bfnm{George}\binits{G.}}
(\byear{2010}).
\btitle{Introducing {M}onte {C}arlo Methods with {R}}.
\bpublisher{Springer},
\blocation{New York}.
\bid{doi={10.1007/978-1-4419-1576-4}, mr={2572239}}
\end{bbook}
%
\bptok{imsref}%
\endbibitem

\bibitem[\protect\citeauthoryear{Rue and Held}{2005}]{RueGMRFBook}
%
\begin{bbook}[mr]
\bauthor{\bsnm{Rue},~\bfnm{H{\aa}vard}\binits{H.}} \AND
\bauthor{\bsnm{Held},~\bfnm{Leonhard}\binits{L.}}
(\byear{2005}).
\btitle{Gaussian {M}arkov Random Fields: Theory and Applications}.
\bseries{Monographs on Statistics and Applied Probability}
\bvolume{104}.
\bpublisher{Chapman \& Hall/CRC},
\blocation{Boca Raton, FL}.
\bid{doi={10.1201/9780203492024}, mr={2130347}}
\end{bbook}
%
\bptok{imsref}%
\endbibitem

\bibitem[\protect\citeauthoryear{Rue, Martino and
Chopin}{2009}]{rue2009approximate}
%
\begin{barticle}[mr]
\bauthor{\bsnm{Rue},~\bfnm{H{\aa}vard}\binits{H.}},
\bauthor{\bsnm{Martino},~\bfnm{Sara}\binits{S.}} \AND
\bauthor{\bsnm{Chopin},~\bfnm{Nicolas}\binits{N.}}
(\byear{2009}).
\btitle{Approximate {B}ayesian inference for latent {G}aussian models
by using integrated nested {L}aplace approximations}.
\bjournal{J. R. Stat. Soc. Ser. B. Stat. Methodol.}
\bvolume{71}
\bpages{319--392}.
\bid{doi={10.1111/j.1467-9868.2008.00700.x}, issn={1369-7412}, mr={2649602}}
\end{barticle}
%
\bptok{imsref}%
\endbibitem

\bibitem[\protect\citeauthoryear{Schr{\"o}dle and
Held}{2011}]{schrodle2011spatio}
%
\begin{barticle}[mr]
\bauthor{\bsnm{Schr{\"o}dle},~\bfnm{Birgit}\binits{B.}} \AND
\bauthor{\bsnm{Held},~\bfnm{Leonhard}\binits{L.}}
(\byear{2011}).
\btitle{Spatio-temporal disease mapping using {INLA}}.
\bjournal{Environmetrics}
\bvolume{22}
\bpages{725--734}.
\bid{doi={10.1002/env.1065}, issn={1180-4009}, mr={2843139}}
\end{barticle}
%
\bptok{imsref}%
\endbibitem

\bibitem[\protect\citeauthoryear{Shaby}{2014}]{shaby2014open}
%
\begin{barticle}[mr]
\bauthor{\bsnm{Shaby},~\bfnm{Benjamin~A.}\binits{B.~A.}}
(\byear{2014}).
\btitle{The open-faced sandwich adjustment for MCMC using estimating
functions}.
\bjournal{J. Comput. Graph. Statist.}
\bvolume{23}
\bpages{853--876}.
\bid{doi={10.1080/10618600.2013.842174}, issn={1061-8600}, mr={3224659}}
\end{barticle}
%
\bptok{imsref}%
\endbibitem

\bibitem[\protect\citeauthoryear{Sherlock
et~al.}{2015}]{sherlock2013efficiency}
%
\begin{barticle}[mr]
\bauthor{\bsnm{Sherlock},~\bfnm{Chris}\binits{C.}},
\bauthor{\bsnm{Thiery},~\bfnm{Alexandre~H.}\binits{A.~H.}},
\bauthor{\bsnm{Roberts},~\bfnm{Gareth~O.}\binits{G.~O.}} \AND
\bauthor{\bsnm{Rosenthal},~\bfnm{Jeffrey~S.}\binits{J.~S.}}
(\byear{2015}).
\btitle{On the efficiency of pseudo-marginal random walk {M}etropolis
algorithms}.
\bjournal{Ann. Statist.}
\bvolume{43}
\bpages{238--275}.
\bid{doi={10.1214/14-AOS1278}, issn={0090-5364}, mr={3285606}}
\bptnote{check volume, check pages, check year}%
\end{barticle}
%
\bptok{imsref}%
\endbibitem

\bibitem[\protect\citeauthoryear{Silvertown and
Antonovics}{2001}]{silvertown2001integrating}
%
\begin{bbook}[auto:parserefs-M02]
\bauthor{\bsnm{Silvertown},~\bfnm{J.}\binits{J.}} \AND
\bauthor{\bsnm{Antonovics},~\bfnm{J.}\binits{J.}}
(\byear{2001}).
\btitle{Integrating Ecology and Evolution in a Spatial Context: 14th
Special Symposium of the British Ecological Society}
\bvolume{14}.
\bpublisher{Cambridge Univ. Press},
\blocation{Cambridge}.
\end{bbook}
%
\bptok{imsref}%
\endbibitem

\bibitem[\protect\citeauthoryear{Tavar{\'{e}}
et~al.}{1997}]{tavare1997inferring}
%
\begin{barticle}[auto:parserefs-M02]
\bauthor{\bsnm{Tavar{\'{e}}},~\bfnm{S.}\binits{S.}},
\bauthor{\bsnm{Balding},~\bfnm{D.~J.}\binits{D.~J.}},
\bauthor{\bsnm{Griffiths},~\bfnm{R.~C.}\binits{R.~C.}} \AND
\bauthor{\bsnm{Donnelly},~\bfnm{P.}\binits{P.}}
(\byear{1997}).
\btitle{Inferring coalescence times from dna sequence data}.
\bjournal{Genetics}
\bvolume{145}
\bpages{505--518}.
\end{barticle}
%
\bptok{imsref}%
\endbibitem

\bibitem[\protect\citeauthoryear{Taylor and Diggle}{2014}]{taylor2013inla}
%
\begin{barticle}[mr]
\bauthor{\bsnm{Taylor},~\bfnm{Benjamin~M.}\binits{B.~M.}} \AND
\bauthor{\bsnm{Diggle},~\bfnm{Peter~J.}\binits{P.~J.}}
(\byear{2014}).
\btitle{I{NLA} or MCMC? A~tutorial and comparative evaluation for
spatial prediction in log-{G}aussian {C}ox processes}.
\bjournal{J. Stat. Comput. Simul.}
\bvolume{84}
\bpages{2266--2284}.
\bid{doi={10.1080/00949655.2013.788653}, issn={0094-9655}, mr={3223624}}
\bptnote{check volume, check pages, check year}%
\end{barticle}
%
\bptok{imsref}%
\endbibitem

\bibitem[\protect\citeauthoryear{Troyer and
Wiese}{2005}]{PhysRevLett94170201}
%
\begin{barticle}[auto:parserefs-M02]
\bauthor{\bsnm{Troyer},~\bfnm{M.}\binits{M.}} \AND
\bauthor{\bsnm{Wiese},~\bfnm{U.-J.}\binits{U.-J.}}
(\byear{2005}).
\btitle{Computational complexity and fundamental limitations to
fermionic quantum Monte Carlo simulations}.
\bjournal{Phys. Rev. Lett.}
\bvolume{94}
\bpages{170201}.
\end{barticle}
%
\bptok{imsref}%
\endbibitem

\bibitem[\protect\citeauthoryear{Van~Duijn, Gile and
Handcock}{2009}]{van2009framework}
%
\begin{barticle}[auto:parserefs-M02]
\bauthor{\bsnm{Van Duijn},~\bfnm{M.~A.}\binits{M.~A.}},
\bauthor{\bsnm{Gile},~\bfnm{K.~J.}\binits{K.~J.}} \AND
\bauthor{\bsnm{Handcock},~\bfnm{M.~S.}\binits{M.~S.}}
(\byear{2009}).
\btitle{A~framework for the comparison of maximum pseudo-likelihood
and maximum likelihood estimation of exponential family random graph models}.
\bjournal{Social Networks}
\bvolume{31}
\bpages{52--62}.
\end{barticle}
%
\bptok{imsref}%
\endbibitem

\bibitem[\protect\citeauthoryear{Walker}{2011}]{walker2011posterior}
%
\begin{barticle}[mr]
\bauthor{\bsnm{Walker},~\bfnm{Stephen~G.}\binits{S.~G.}}
(\byear{2011}).
\btitle{Posterior sampling when the normalizing constant is unknown}.
\bjournal{Comm. Statist. Simulation Comput.}
\bvolume{40}
\bpages{784--792}.
\bid{doi={10.1080/03610918.2011.555042}, issn={0361-0918}, mr={2783887}}
\end{barticle}
%
\bptok{imsref}%
\endbibitem

\bibitem[\protect\citeauthoryear{Walker}{2014}]{walker2014bayesian}
%
\begin{barticle}[mr]
\bauthor{\bsnm{Walker},~\bfnm{Stephen~G.}\binits{S.~G.}}
(\byear{2014}).
\btitle{A {B}ayesian analysis of the {B}ingham distribution}.
\bjournal{Braz. J. Probab. Stat.}
\bvolume{28}
\bpages{61--72}.
\bid{doi={10.1214/12-BJPS193}, issn={0103-0752}, mr={3165428}}
\end{barticle}
%
\bptok{imsref}%
\endbibitem

\bibitem[\protect\citeauthoryear{Wang and Landau}{2001}]{wang2001efficient}
%
\begin{barticle}[auto:parserefs-M02]
\bauthor{\bsnm{Wang},~\bfnm{F.}\binits{F.}} \AND
\bauthor{\bsnm{Landau},~\bfnm{D.~P.}\binits{D.~P.}}
(\byear{2001}).
\btitle{Efficient, multiple-range random walk algorithm to calculate
the density of states}.
\bjournal{Phys. Rev. Lett.}
\bvolume{86}
\bpages{2050}.
\end{barticle}
%
\bptok{imsref}%
\endbibitem

\bibitem[\protect\citeauthoryear{Welling and Teh}{2011}]{welling2011bayesian}
%
\begin{binproceedings}[auto:parserefs-M02]
\bauthor{\bsnm{Welling},~\bfnm{M.}\binits{M.}} \AND
\bauthor{\bsnm{Teh},~\bfnm{Y.~W.}\binits{Y.~W.}}
(\byear{2011}).
\btitle{Bayesian learning via stochastic gradient Langevin dynamics}.
In \bbooktitle{Proceedings of the 28th International Conference on
Machine Learning}
\bpages{681--688}.
\bpublisher{Omnipress}, \blocation{Madison, WI}.
\end{binproceedings}
%
\bptok{imsref}%
\endbibitem

\bibitem[\protect\citeauthoryear{Zhang et~al.}{2012}]{zhang2012continuous}
%
\begin{bincollection}[auto:parserefs-M02]
\bauthor{\bsnm{Zhang},~\bfnm{Y.}\binits{Y.}},
\bauthor{\bsnm{Ghahramani},~\bfnm{Z.}\binits{Z.}},
\bauthor{\bsnm{Storkey},~\bfnm{A.~J.}\binits{A.~J.}} \AND
\bauthor{\bsnm{Sutton},~\bfnm{C.~A.}\binits{C.~A.}}
(\byear{2012}).
\btitle{Continuous relaxations for discrete Hamiltonian Monte Carlo}.
In \bbooktitle{Advances in Neural Information Processing Systems}
\bvolume{4}
\bpages{3194--3202}.
\end{bincollection}
%
\bptok{imsref}%
\endbibitem

\bibitem[\protect\citeauthoryear{Zhou and Schmidler}{2009}]{zhou2009bayesian}
%
\begin{bmisc}[auto:parserefs-M02]
\bauthor{\bsnm{Zhou},~\bfnm{X.}\binits{X.}} \AND
\bauthor{\bsnm{Schmidler},~\bfnm{S.}\binits{S.}}
(\byear{2009}).
\bhowpublished{Bayesian parameter estimation in Ising and Potts
models: A comparative study with applications to protein modeling.
Technical report}.
\end{bmisc}
%
\bptok{imsref}%
\endbibitem

\end{thebibliography}
\end{document}